\documentclass[journal]{IEEEtran}

\usepackage[utf8]{inputenc}
\usepackage[T1]{fontenc}
\usepackage{url}
\usepackage{ifthen}
\usepackage{cite}
\usepackage[cmex10]{amsmath} 

\usepackage{graphicx}
\usepackage{mathrsfs}
\usepackage{amsfonts}

\usepackage{amsthm}
\usepackage{dsfont}
\theoremstyle{remark}
\newtheorem{definition}{Definition}

\newtheorem{example}{Example}
\newtheorem{theorem}{Theorem}[section]
\newtheorem{lemma}{Lemma}[section]

\newtheorem{remark}{Remark}
\newtheorem{assumption}{Assumption}
\usepackage{color}
\usepackage{comment}
\usepackage{bm}
\usepackage{algorithm}
\usepackage{algorithmic}

\newcommand{\R}{\ensuremath{\mathbb{R}}}

\def\0{\mathbf{0}}
\def\1{\mathds{1}}

\def\bbf{\mathbf{b}}
\def\Abf{\mathbf{A}}

\def\Bbf{\mathbf{B}}

\def\Dbf{\mathbf{D}}
\def\e{\mathbf{e}}
\def\ep{\boldsymbol{\epsilon}}

\def\Ep{\mathbb{E}}

\def\Gbf{\mathbf{G}}

\def\I{\mathbf{I}}

\def\Lbf{\mathbf{L}}

\def\Mbf{\mathbf{M}}
\def\N{\mathcal{N}}

\def\Pbf{\mathbf{P}}
\def\Qbf{\mathbf{Q}}
\def\Rbf{\mathbf{R}}
\def\Rcal{\mathcal{R}}
\def\rbf{\mathbf{r}}

\def\Sbf{\mathbf{S}}

\def\Ubf{\mathbf{U}}
\def\u{\mathbf{u}}
\def\vv{\mathbf{v}}

\def\Vbf{\mathbf{V}}

\def\w{\mathbf{w}}
\def\Wbf{\mathbf{W}}
\def\y{\mathbf{y}}
\def\x{\mathbf{x}}

\def\Xbf{\mathbf{X}}

\def\z{\mathbf{z}}
\def\Zbf{\mathbf{Z}}

\def\one{{\bf 1}}
\def\zero{{\bf 0}}

\def\trace{\text{trace}}

\newcommand{\twonorm}[1]{\left\| #1 \right\|}
\newcommand{\mat}[1]{\text{mat}\left( #1 \right)}
\newcommand{\vect}[1]{\text{vec}\left( #1 \right)}

\newcommand{\wt}[1]{\widetilde{#1}}

\begin{document}

\title{Coded Iterative Computing using Substitute Decoding}

\author{%
   \IEEEauthorblockN{Yaoqing Yang, Malhar Chaudhari, Pulkit Grover, Soummya Kar}
\thanks{Some parts of this paper will be presented at the 2018 IEEE International Symposium on Information Theory (ISIT) \cite{yang2018coding}. This work is supported by NSF ECCS-1343324, NSF CCF-1350314 (NSF CAREER) for Pulkit Grover, NSF ECCS-1306128, NSF CCF-1513936, the Bertucci Graduate Fellowship for Yaoqing Yang, and by Systems on Nanoscale Information fabriCs (SONIC), one of the six SRC STARnet Centers, sponsored by MARCO and DARPA.

Y. Yang, M. Chaudhari, P. Grover and S. Kar are with the Department of Electrical and Computer Engineering, Carnegie Mellon University, Pittsburgh, PA, 15213, USA. Email: \{yyaoqing, mschaudh, pgrover, soummyak\}@andrew.cmu.edu}}

\maketitle

\begin{abstract}
In this paper, we propose a new coded computing technique called ``substitute decoding'' for general iterative distributed computation tasks. In the first part of the paper, we use PageRank as a simple example to show that substitute decoding can make the computation of power iterations solving PageRank on sparse matrices robust to erasures in distributed systems. For these sparse matrices, codes with dense generator matrices can significantly increase storage costs and codes with low-density generator matrices (LDGM) are preferred. Surprisingly, we show through both theoretical analysis and simulations that when substitute decoding is used, coded iterative computing with extremely low-density codes (2 or 3 non-zeros in each row of the generator matrix) can achieve almost the same convergence rate as noiseless techniques, despite the poor error-correction ability of LDGM codes. In the second part of the paper, we discuss applications of substitute decoding beyond solving linear systems and PageRank. These applications include (1) computing eigenvectors, (2) computing the truncated singular value decomposition (SVD), and (3) gradient descent. These examples show that the substitute decoding algorithm is useful in a wide range of applications.
\end{abstract}

\section{INTRODUCTION}

Iterative computation is one of the most fundamental tools in statistics and numerical linear algebra. Optimization problems can often be solved using iterative methods such as gradient descent and its extensions. Spectral analysis on large and sparse matrices, such as computing eigenvectors and the singular value decomposition, can also be conducted in an iterative way \cite{rutishauser1969computational,stewart1976simultaneous,golub2012matrix,berry2006parallel,halko2011finding,kempe2008decentralized,ji2016apache}. Another typical example of iterative computing is the power-iteration method that repeatedly multiplies the intermediate result with a matrix until convergence, which is useful in a variety of applications such as PageRank \cite{page1999pagerank,haveliwala2002topic}, semi-supervised learning \cite{zhu2003semi} and clustering \cite{lin2010power}.

In this work, we utilize error correcting codes to make distributed iterative computing robust to system noise such as stragglers and erasures. Researchers from the coded computing field have investigated a wide range of applications, including matrix multiplications, distributed frameworks, and large-scale machine learning algorithms \cite{lee2016speeding,tandon2016gradient,dutta2016short,ferdinand2016anytime,li2018fundamental,yu2017polynomial,karakus2017straggler,jeongcoded,fahim2017optimal,bitar2017minimizing,attia2018near,aktas2017straggler,woolsey2018new,li2015coded,duttaISIT2017,YaoqingISTC,wang2018fundamental,maityrobust}. Our previous work \cite{yang2017coded} considered coded computing for power iterations when many inverse problem instances of the form $\Mbf\x_i=\bbf_i,i=1,2,\ldots,m$ are computed in parallel. However, a more common setting is the computation of a \emph{single} inverse problem $\Mbf\x=\bbf$ when the linear system matrix $\Mbf$ is large and sparse and thus cannot fit in the memory of a single machine, and thus distributed computing is necessary. Existing results in coded computing typically use dense generator matrices, such as those of MDS codes, to ensure good error-correcting capability. However, dense encoding of the sparse matrix $\Mbf$ can significantly increase the number of non-zero entries, and hence can increase communication, storage and computational costs. In fact, dense encoding using MDS codes makes the sparse problem completely non-sparse. In this work, we use codes with sparse generator matrices instead (low-density generator matrices, or LDGM). However, despite the bad error-correcting ability of LDGM codes, we show that LDGM codes can be made surprisingly efficient in maintaining the convergence rate of iterative computing. The key is to use a novel decoding algorithm that we call ``substitute decoding''.

The ``substitute decoding'' method relies on the fact that the intermediate result in the iterative computing $\x_{t+1}=f(\x_t)$ gradually converges to the fixed point or the optimum point, so $\x_{t+1}$ and $\x_t$ gradually become close to each other when $t$ increases. Using this property, the substitute decoding method works by extracting the largest amount of available information from partial coded results in the computation of $\x_{t+1}=f(\x_t)$, and substituting the complementary unknown information by the available side information $\x_t$ from the previous step. In other words, instead of computing the exact result of $\x_{t+1}=f(\x_t)$, we compute a combined version $\x_{t+1}=\text{Proj}_1 [f(\x_t)]+\text{Proj}_2 [\x_t]$, where $\text{Proj}_1$ represents the projection onto the space of available information from partial decoding, and $\text{Proj}_2$ represents the orthogonal projection of $\text{Proj}_1$. This is useful in coded computing with LDGM codes because even if the exact result is not available, we can obtain a partial result and use the side information $\x_t$ to compensate the information loss. As we show in our main theorems (Theorem~\ref{thm:converge} and Theorem \ref{thm:converge_col}), substitute decoding can reduce error by a multiplicative factor $\delta$ that drops to 0 when the partial generator matrix is close to full-rank. More specifically, $\delta$ is linear in the rank of the partial generator matrix formed by the linear combinations from non-erased workers. The convergence rate of noiseless computation can be achieved when $\delta$ is small. This property is essential in using LDGM codes for coding sparse data because the partial generator matrix is often not full-rank but close to full-rank, which makes $\text{Proj}_1$ approximately equal to the identity projection, and makes $\x_{t+1}=\text{Proj}_1 [f(\x_t)]+\text{Proj}_2 [\x_t]$ close to $f(\x_t)$. As we show in our simulations (see Section~\ref{sec:simulation_all}), even for sparse generator matrices with only 2 non-zero entries in each row, coded iterative computing works significantly better than replication-based or uncoded iterative computing in convergence rate when the results from a constant fraction of workers are erased. When there are 3 non-zeros in each row, the convergence rate of noiseless iterative computing (the information limit) can be approximately achieved by coded computing.

In the first part of the paper (see Section~\ref{sec:II} to Section~\ref{sec:decoding}), we use PageRank (see Section~\ref{sec:PageRank}) as a simple example of iterative computing and introduce the substitute decoding method for different types of data splitting. We use the power-iteration method to compute PageRank which essentially computes the principal eigenvector of a sparse matrix. In the second part of the paper, we show that substitute decoding can be applied to more iterative computing problems beyond PageRank and linear systems. For example, we show that substitute decoding can make the orthogonal-iteration method \cite{rutishauser1969computational,stewart1976simultaneous,golub2012matrix,berry2006parallel,halko2011finding,kempe2008decentralized,ji2016apache,ankur2018rateless} robust to erasure-type failures, which can be applied to computing more than one eigenvectors (see Section~\ref{sec:spectral_clustering}) and the truncated singular value decomposition (see Section~\ref{sec:pca}). We also show that substitute decoding can be applied to the computation of gradient descent with or without sparse data (see Section~\ref{sec:gdcoding}) and can improve on existing techniques when an extremely sparse encoding matrix is used. In Section \ref{sec:simulation_all}, we compare substitute-decoding-based methods against other baseline methods for the same communication cost. In Section \ref{sec:proof_all}, we provide the proofs for all the theorems.

Now, we summarize the contributions of this work.
\begin{itemize}
\item We design the substitute decoding method and show it can make coded iterative computing using LDGM codes approximately achieve the convergence rate of noiseless computation, in both theory and simulation.
\item We apply substitute decoding to many applications, such as large-scale eigendecomposition and singular-value-decomposition that have applications in graph spectral clustering, principal component analysis on sparse matrices and anomaly detection.
\end{itemize}

\subsection{Related works}

A closely related line of works is coding for gradient-descent-type algorithms \cite{karakus2017straggler,tandon2016gradient,karakus2017encoded, raviv2017gradient,charles2017approximate,ye2018communication,halbawi2017improving}. For gradient coding on nonlinear gradient functions, a sparse code is necessary because the storage overhead and the computational cost are directly related to the number of non-zeros in the code. Also, since in real applications the data are often sparse, the coding matrix needs to be sparse as well. For example, the necessity of using sparse codes for sparse data and the idea of storing uncoded relevant sparse data at each worker as specified by the sparse encoding matrix first appear in \cite{karakus2017straggler}. Comparing to coding for gradient descent, our substitute decoding method focuses on a more general framework of iterative computing beyond gradient descent. For example, power iterations and the more general Jacobi iterations and orthogonal iterations (see Section \ref{sec:spectral_clustering}) are not gradient-descent-type methods and the coding techniques used in this paper are also different. Moreover, if the exact gradient is required at each iteration, there is a tight lower bound on the number of non-zeros and the bound is linear in the number of erasures \cite{tandon2016gradient}. However, we focus on the regime where the encoding matrix is extremely sparse, which means the number of ones in each row of the encoding matrix is a constant (2 to 3) and does not increase with the number of erasures. However, we do notice interesting intellectual connections between substitute decoding and recent works on approximate gradient coding \cite{raviv2017gradient,charles2017approximate} because in these works remaining inaccuracy is also allowed in the computation of gradients, and the lower bound on the number of non-zeros in \cite{tandon2016gradient} can be relaxed. The main difference between approximate gradient coding and substitute decoding, when applied specifically to gradient descent, is that the substitute decoding method tries to recover the partial results from all workers, instead of focusing on computing only the SUM function of the results. From this angle, substitute decoding provides a way to introduce momentum into the system using partial gradients computed from previous iterations. Some other works in coded computing focus on sparsifying the encoded matrix (not the encoding matrix) for coded matrix multiplications \cite{dutta2016short,suh2017matrix}, but these schemes are generally not available for sparsifying data matrices that are already sparse, such as the graph adjacency matrix for PageRank. It is also well-known that random sparsification on graphs may reduce the computational complexity while maintaining the eigenvectors and the spectrum \cite{korada2011gossip,spielman2011graph}. Although these sparsification techniques are generally not designed for erasures on the entries of the intermediate result (which is equivalent to erasing some entire row blocks of the graph matrix), they can be used together with the substitute decoding technique to further reduce the computational complexity, provided that the same sparsification is applied to replicas of the same row block.

Some concurrent works also focus on coded computing for graph analytics \cite{prakash2018coded} and sparse matrix multiplications \cite{wang2018coded,wang2018fundamental,ankur2018rateless}. These works focus on single-shot computing instead of iterative computing. For iterative computational tasks, one has to optimize the system performance from the entire convergence process. For example, although the substitute decoding method is not optimal at each iteration (it does not even provide the complete decoding result in a single iteration), it has the desired convergence rate in the entire iterative process. We notice that the multistage coded computation described by a directed acyclic graph has been discussed in an earlier work \cite{li2016coded} under the Map-Reduce framework. The importance of the convergence-rate for iterative computing in the context of computing with stragglers has also been well-recognized by recent works \cite{dutta2018slow,karakus2017straggler}.

Another intellectually related line of works is on coding for \emph{noisy computing} in which the circuit logic units can have faults \cite{Had_TIT_05, Yaz_TC_01,Var_TIT_11, Chi_ITW_07,taylor1968reliable,Pip_FOC_85,Pip_TIT_91}. One idea in these works was that even if errors cannot be eliminated completely during each iteration, the error may not diverge if sufficient (low-complexity) decoding is implemented after each iteration. This idea of iterative error suppression has been applied recently in our previous works to computing linear transforms entirely out of unreliable components \cite{yang2017computing} and logistic regression with erroneous units \cite{yang2016fault}. The fundamental limits of iterative error suppression have also been studied extensively \cite{Eva_TIT_99} and are closely related to the concept of \emph{information dissipation} \cite{polyanskiy2016dissipation,yang2017rate}. Compared to these works, we further show that even if the code that we use is very sparse (LDGM), we can still manage to maintain the convergence rate in the noiseless case when \emph{sufficient errors} are corrected during each iteration. This is fundamentally different from previous works because substitute decoding relies on the advantage of iterative convergence itself to make a sparse code useful, while all the mentioned previous works use dense codes. Note that the error-correcting ability of sparse codes has been studied in control theory as well in the context of observability in the presence of malicious attacks \cite{shoukry2016event,chen2015cyber}, but when the linear system matrix is 2$s$-sparse the allowable number of failures is only $s$. In our work, the number of tolerable erasures is linear in the number of workers, even when the code is 2$s$-sparse. This is similar to the conclusion of \cite{charles2017approximate} in the study of gradient coding.

\section{Problem Formulation}\label{sec:II}
\subsection{Preliminaries on PageRank and the Power-Iteration Method}\label{sec:PageRank}

We use the PageRank problem \cite{page1999pagerank,haveliwala2002topic} as an example to introduce substitute decoding and delegate the discussion on more iterative computing tasks to Section \ref{sec:extension}. The PageRank and the more general personalized PageRank problem aim to measure the importance score of the nodes on a graph by solving the linear problem $\x=c \rbf+(1-c)\Abf\x$, where $c=0.15$ is a constant, $N$ is the number of nodes, $\Abf$ is the column-normalized\footnote{By column-normalized, we mean the columns of $\Abf$ are normalized so that each column has sum 1.} adjacency matrix and $\rbf\in \R^N$ represents the preference of different users or topics. The classical way to solve PageRank is to use the power-iteration method, which iterates $\x_{t+1}=c \rbf+(1-c)\Abf\x_t$ until convergence. If we define $\Bbf=(1-c)\Abf$ and $\y=c\rbf$. Then, we obtain the general form of power iterations:
\begin{equation}\label{eqn:power_iteration}
\x_{t+1}=\y+\Bbf\x_t.
\end{equation}
The condition that \eqref{eqn:power_iteration} converges to the true solution $\x^*$ is that the spectral radius $\rho(\Bbf)<1$. For the PageRank problem, $\rho(\Bbf)=1-c$ and \eqref{eqn:power_iteration} always converges to $\x^*$.

\begin{figure}
  \centering
  \includegraphics[scale=0.3]{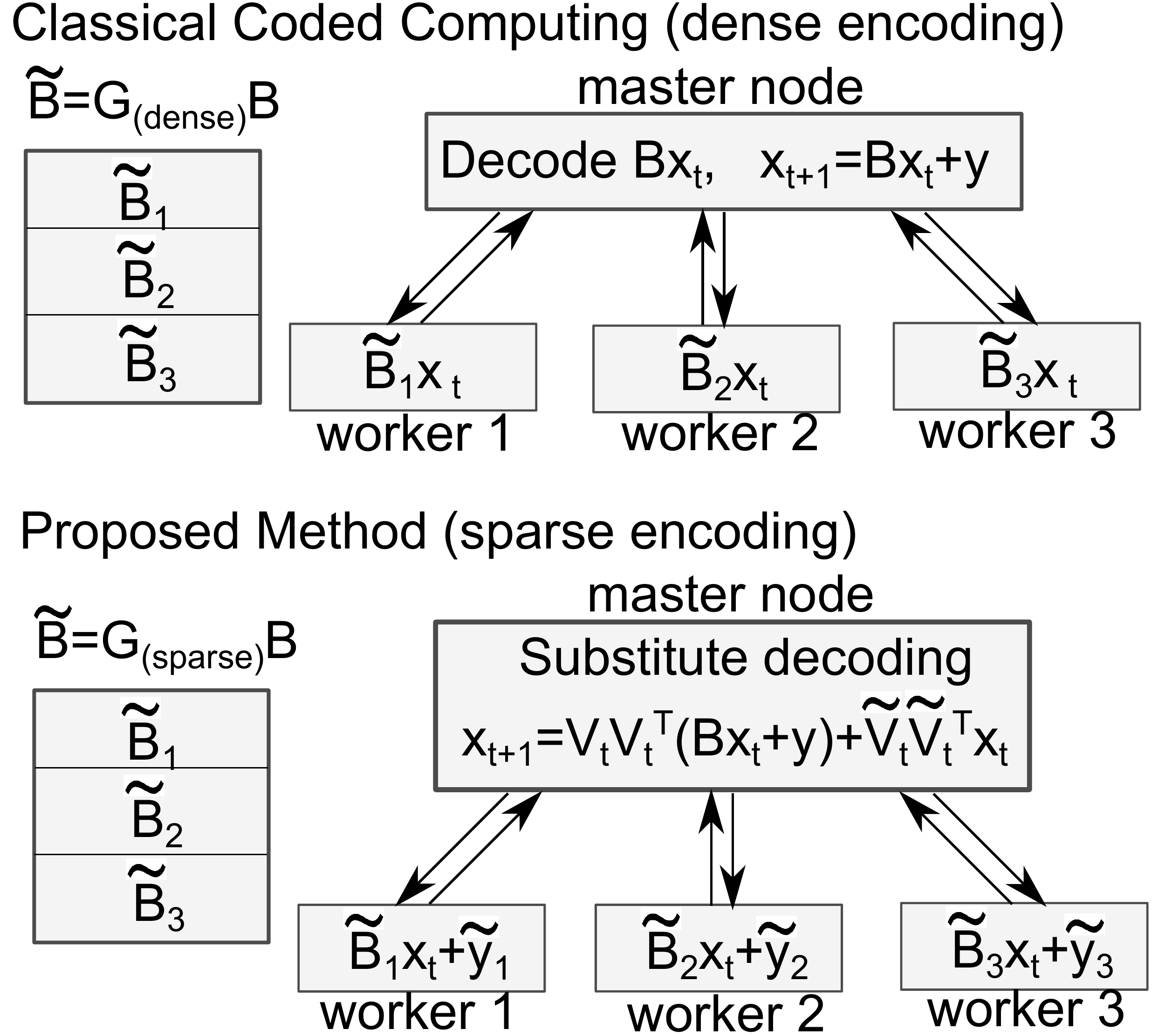} \\
  \caption{This shows the comparison between the existing works on coded computing and the proposed coded computing technique for the computation of power iterations in the row-splitting case. In the proposed method, we only show the scalar version as mentioned in Remark~\ref{rem:scalar}.} \label{fig:comparison}
\end{figure}
\subsection{Noiseless Distributed Computing of Power Iterations}

When the size of the linear system matrix $\Bbf$ is too large to fit in the memory of a single machine, the computation of \eqref{eqn:power_iteration} is performed distributedly. The most straightforward way is to partition $\Bbf$ into several blocks and store them in the memory of several workers. In this section, we describe three types of data splitting, namely row-wise splitting, column-wise splitting and SUMMA splitting (i.e., both row and column).
\subsubsection{Row-wise Splitting}\label{sec:row-splitting}
We split the linear system matrix $\Bbf$ into several row blocks and store them in the memory of several workers. Denote the number of workers by $P$, and this is also the number of row blocks. At the beginning of the $t$-th iteration, a master node sends the current result $\x_t$ to all workers. Then, the worker that has the row block $\Bbf_i$ computes $\Bbf_i \x_t$ and sends it back to the master node. At the end of the iteration, the master node concatenates all the results $\Bbf_i\x_t,i=1,2,\ldots,P$ from the $P$ workers to obtain $\Bbf \x_t$, and computes $\x_{t+1}=\Bbf \x_t+\y$.
\subsubsection{Column-wise Splitting}
We split the linear system matrix $\Bbf$ into several column blocks and store them in the memory of several workers. Again denote the number of workers by $P$, and this is also the number of column blocks. At the beginning of the $t$-th iteration, a master node breaks the current result $\x_t$ into $P$ subvectors of the same length $N/P$ \footnote{When $N$ is not divisible by $P$, we can add some zero columns to the matrix $\Bbf$} and sends each subvector $\x_t^i$ to the $i$-th worker. Then, the $i$-th worker computes $\Bbf_i \x_t^i$ and sends it back to the master node. At the end of the iteration, the master node computes the sum of all the results $\Bbf_i\x_t^i,i=1,2,\ldots,P$ from the $P$ workers to obtain $\Bbf \x_t$, and computes $\x_{t+1}=\Bbf \x_t+\y$.
\subsubsection{SUMMA Splitting}
The SUMMA matrix splitting is designed for matrix-matrix multiplications \cite{van1997summa} to reduce the communication complexity, but it can be applied to matrix-vector multiplications as well. In SUMMA splitting, we split the linear system matrix $\Bbf$ both row-wise and column-wise into $\sqrt{P}\times \sqrt{P}$ blocks $\Bbf_{ij},i=1,\ldots,\sqrt{P},j=1,\ldots,\sqrt{P}$ and store them in the memory of $P$ workers. At the beginning of the $t$-th iteration, a master node breaks the current result $\x_t$ into $\sqrt{P}$ subvectors of the same length $N/\sqrt{P}$ and sends each subvector $\x_t^j$ to all the workers that have blocks $\Bbf_{ij}$, where $i=1,2,\ldots \sqrt{P}$ and $j$ is fixed. Then, the $(i,j)$-th worker computes $\Bbf_{ij} \x_t^j$ and sends it back to the master node. After that, the master node computes the sum of all the results $\Bbf_{ij} \x_t^j,j=1,2,\ldots,\sqrt{P}$ for a fixed $i$ to obtain a subvector with index $i$, and concatenates all the subvectors for $i=1,2,\ldots,\sqrt{P}$ to obtain $\Bbf \x_t$. Finally, the master node computes $\x_{t+1}=\Bbf \x_t+\y$.

Since the data matrix $\Bbf$ can have different types of splitting methods, the coded computing technique should be able to apply to different splitting situations as well. In Section \ref{sec:coding_preliminaries2}, we show that our coded computing technique can indeed be applied to different matrix splitting situations.

\subsection{Preliminaries and Notation on Coded Computing}\label{sec:coding_preliminaries1}

Error correcting codes can help make distributed computing robust to stragglers and erasures. We first present the direct application of coded computing to the power iteration \eqref{eqn:power_iteration} with row-wise splitting, and point out a drawback of it. As shown in the upper part of Fig.~\ref{fig:comparison}, if the common idea of coded computing is applied to row-wise splitting, $\Bbf_{N\times N}$ is partitioned into $k$ row blocks and linearly combined into $P>k$ row blocks $\wt\Bbf_i,i=1,2,\ldots,P$ using a $(P,k)$ linear code with a generator matrix $\Gbf$ of size $P\times k$. Each encoded block is stored at one of the $P$ workers. For simplicity, we assume that $N$ is divisible by $k$, and denote the number of rows in each row block by $b=\frac{N}{k}$. Then, encoding can be written as
\begin{equation}\label{eqn:B_encoding}
\wt\Bbf=(\Gbf_{P\times k}\otimes \I_b)\Bbf,
\end{equation}
where the Kronecker product is because we encode row blocks.
\begin{example}\label{ex:1}
We show an example of coded computing \cite{lee2016speeding}. This example will be mentioned many times throughout the paper. Suppose $\Bbf$ is partitioned into $\Bbf_{N\times N}=\left[\begin{matrix}
\Bbf_1\\
\Bbf_2
\end{matrix}\right]$ and encoded into $\wt\Bbf_{\frac{3}{2}N\times N}=\left[\begin{matrix}
\Bbf_1\\
\Bbf_2\\
\Bbf_1+\Bbf_2
\end{matrix}\right]$. The generator matrix is $\Gbf_{3\times 2}=\left[\begin{matrix}
1&0\\
0&1\\
1&1
\end{matrix}\right]$. The number of row blocks in $\Bbf$ is $k=2$. The number of workers is $P=3$ (which is also the code length). The number of rows in each row block is $b=\frac{N}{2}$. The encoding can indeed by written as \eqref{eqn:B_encoding}, because
\[\left[\begin{matrix}
\Bbf_1\\
\Bbf_2\\
\Bbf_1+\Bbf_2
\end{matrix}\right]=\left[\begin{matrix}
\I_b & \\
 & \I_b\\
\I_b &\I_b
\end{matrix}\right]\left[\begin{matrix}
\Bbf_1\\
\Bbf_2
\end{matrix}\right]=\left(\left[\begin{matrix}
1&0\\
0&1\\
1&1
\end{matrix}\right]\otimes \I_b \right) \left[\begin{matrix}
\Bbf_1\\
\Bbf_2
\end{matrix}\right].
\]
\end{example}

At each iteration, the $i$-th worker computes $\wt\Bbf_i \x_t$ and the master node concatenates all the results to obtain $\wt\Bbf \x_t$, which is a coded version of $\Bbf\x_t$.

Now, we define two operations (shown in Fig.~\ref{fig:vec_mat}). Denote by $\vv=\vect{\Xbf}$ the operation to vectorize the matrix $\Xbf$ into the concatenation of its transposed rows, and denote by $\Xbf=\mat{\vv}$ the operation to partition the column vector $\vv$ into small vectors and stack the transposed small vectors into the rows of $\Xbf$. We always partition the vector $\vv$ into smaller ones of length $b=\frac{N}{k}$ which represents the row-block size at each worker. Then, it is straightforward to show that any operation of the form $\x=(\Abf\otimes \I_b)\cdot \vv$ can be rewritten in a compact form $\x=\vect{\Abf\mat{\vv}}$. Therefore, from \eqref{eqn:B_encoding}, the obtained results at the master node is
\begin{equation}
\wt\Bbf\x=(\Gbf_{P\times k}\otimes \I_b)\Bbf\x=\vect{\Gbf\mat{\Bbf\x}}.
\end{equation}
The matrix-version of $\wt\Bbf\x$ is hence $\Gbf\mat{\Bbf\x}$, which means each column of the matrix-version of $\wt\Bbf\x$ is a codeword. In the presence of stragglers or erasures, the coded results $\Gbf\mat{\Bbf\x}$ would lose some of its rows, and the decoding can be done in a parallel fashion on each column. There are $b$ columns in $\Gbf\mat{\Bbf\x}$. Thus, the decoding complexity is $b\N_\text{dec}$, where $\N_\text{dec}$ is the complexity of decoding a single codeword.

\begin{figure}
  \centering
  \includegraphics[scale=0.40]{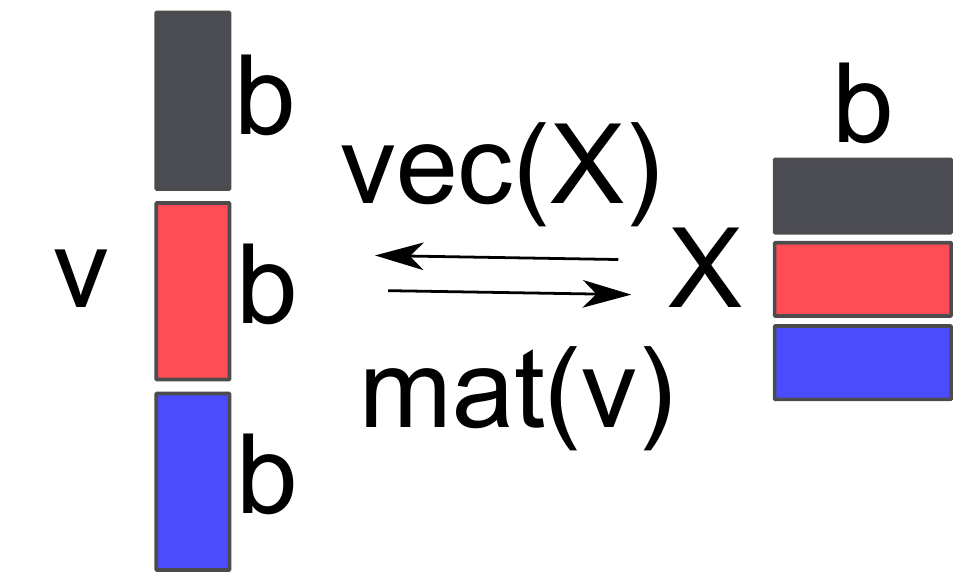}\\
  \caption{An illustration on the $\vect{\cdot}$ and the $\mat{\cdot}$ operations.} \label{fig:vec_mat}
\end{figure}

A main drawback of the above method is that the generator matrix $\Gbf$ is usually dense, such as MDS codes or random Gaussian codes. However, in practice, the system matrix $\Bbf$ is often sparse, as in the PageRank problem. Therefore, using a dense $\Gbf$ may significantly increase the number of non-zeros of the sparse linear system matrix. For example, if the generator matrix $\Gbf$ has 20 non-zeros in each row, it means that the submatrices $\wt\Bbf_i$ stored at each worker can have at most 20 times larger size than the uncoded case if the matrix $\Bbf$ is sparse. In this paper, the key question that we would like to ask is that what is the right coding algorithm to use if we have a tight constraint on the sparsity of the encoding matrix (such as only 2 or 3 non-zeros in each row of the encoding matrix), and whether any positive results can be achieved under such constraint. Notice that in some applications such as computing the gradient, if the exact gradient needs to be computed, the sparsity of the encoding matrix has to grow linearly with the number of erasures \cite{tandon2016gradient}.

In section~\ref{sec:decoding_alg}, we show that $\Gbf$ can actually be very sparse (such as two ones in each row), while the iterative computing (such as power iterations) can still remain robust to a linear number of stragglers in $P$. This result is extended in Section \ref{sec:extension} to general iterative computing problems.

\subsection{Preliminaries on the Proposed Technique}\label{sec:coding_preliminaries2}

\subsubsection{Row-wise splitting}
In our technique for row-wise splitting, similar to standard coded computing, the linear system matrix $\Bbf$ is partitioned into $k$ row blocks and encoded into $P$ row blocks using a $(P,k)$ code with rate $R=\frac{k}{P}$. Each encoded row block is stored at one worker. We now state an important difference in our code: \emph{at each iteration, we use a different generator matrix $\Gbf^{(t)}$, but its sparsity pattern remains the same across iterations}. In Example \ref{ex:1}, the generator matrix $\Gbf=\left[\begin{matrix}
1&0\\
0&1\\
1&1
\end{matrix}\right]$. So at each iteration $t$, we have a different generator matrix $\Gbf^{(t)}=\left[\begin{matrix}
g_{11}^t&0\\
0&g_{22}^t\\
g_{31}^t&g_{32}^t
\end{matrix}\right]$. We choose each non-zero $g_{ij}^t$ to be a standard Gaussian r.v., and all of these r.v.s are independent of each other. The fixed sparsity pattern $\Gbf$ determines which (sparse) row blocks of the uncoded matrix $\Bbf$ are stored at each worker. In particular, $\Bbf_j,j=1,\ldots,k$ is stored at worker-$i$, $i=1\ldots P$, when $\Gbf_{i,j}=1$. In Example \ref{ex:1}, $\Bbf_1$ is stored at worker-1 and worker-3, and $\Bbf_2$ is stored at worker-2 and worker-3. However, instead of precomputing the encoded submatrices $\wt\Bbf_i$ as in \eqref{eqn:B_encoding}, the $i$-th worker just stores its required row blocks in $\Bbf$, because the code is time-varying. Similarly, we also partition the vector $\y$ into $k$ subvectors of length $b$ and store them in the $P$ workers in the exactly same fashion as $\Bbf$. At the $t$-th iteration, worker-$i$ computes $(\Gbf^{(t)})_\text{$i$-th row}\mat{\Bbf\x_t+\y}$. In Example \ref{ex:1}, this means that worker-3 stores $\Bbf_1$ and $\Bbf_2$ at the local memory. At the $t$-th iteration, it computes $\Bbf_1\x_t+\y_1$ and $\Bbf_2\x_t+\y_2$ and encodes them to $g_{31}(t)(\Bbf_1\x_t+\y_1)+g_{32}(t)(\Bbf_2\x_t+\y_2)$ using the linear coefficients in $\Gbf^{(t)}$. Since the sparsity pattern is fixed, although the code is time-varying, stored data blocks at each worker remain the same.

At each iteration, a random fraction $\epsilon$ of the workers fail to send their results back due to either erasures (packet losses) or stragglers (the communications with slow workers are discarded to save time). Then, at the master node, available results are $\Gbf_s^{(t)}\mat{\Bbf \x_t+\y}$, where $\Gbf_s^{(t)}$ is the submatrix of $\Gbf_{s}$ formed by the linear combinations at the non-erased workers. We call $\Gbf_s^{(t)}$ a ``partial generator matrix''. In existing works on coded computing, if a dense Vandermonde-type code is used, the desired result $\Bbf \x_t+\y$ can be decoded from $\Gbf_s^{(t)}\mat{\Bbf \x_t+\y}$ if $1-\epsilon>R$, because any square submatrix $\Gbf_s^{(t)}$ of a Vandermonde matrix is invertible. However, if $\Gbf$ is extremely sparse, even if $\epsilon$ is very small, it is possible that $\Bbf \x_t+\y$ cannot be decoded because $\Gbf_s^{(t)}$ may not be invertible.

\subsubsection{Column-wise splitting}\label{sec:col_part_preliminary}

In the column-wise splitting, the matrix $\Bbf$ is partitioned into $k$ column blocks $\Bbf_j,j=1,\ldots k$. Similar to the row-wise splitting case, we do not encode these column blocks explicitly in preprocessing because we are going to use a time-varying code. Instead, we distribute these column blocks according to the (fixed) sparsity pattern matrix $\Gbf$. In particular, $\Bbf_j,j=1,\ldots,k$ is stored at worker-$i$, $i=1\ldots P$, when $\Gbf_{i,j}=1$. At the $t$-th iteration, the master node partitions the vector $\x_t$ into $k$ sub-vectors $\x_t^j$ and sends $\x_t^j$ to all workers that have the column block $\Bbf_j$ in memory. Then, the $i$-th worker computes the linear combination
\begin{equation}\label{eqn:worker_vertical}
\w_t^i=\sum_{j=1}^k g_{ij}^t \Bbf_j \x_t^j,
\end{equation}
and sends back the result. Notice that the communication overhead of sending multiple $\x_t^j$'s instead of one to each worker is larger than in the noiseless case but this overhead is small because (1) the pattern matrix $\Gbf$ is very sparse, and (2) the vectors $\x_t^j$ are of small length $N/k$ compared to the vector $\Bbf_j\x_t^j$ of length $N$, so the communication cost is dominated by the communication from the workers to the master node. If no erasures occur, the master node receives
\begin{equation}\label{eqn:results_from_workers}
\Wbf_t:=[\w_t^1,\w_t^2,\ldots,\w_t^P]=[\Bbf_1\x_t^1,\Bbf_2\x_t^2,\ldots,\Bbf_k\x_t^k]\cdot (\Gbf^{(t)})^\top,
\end{equation}
and decodes $\Bbf\x_t = \sum_{j=1}^k\Bbf_j\x_t^j$ from these column vectors by inverting the matrix $\Gbf^{(t)}$. It then updates $\x_{t+1}$ based on the estimate of $\Bbf\x_t$. However, since the matrix $\Gbf$ is sparse, when there are erasures, the submatrix $\Gbf_s^{(t)}$ may again be not invertible similar to the case of row-wise splitting.

\subsubsection{SUMMA-type splitting}\label{sec:SUMMA}

Although an entangled way of applying coded computing to the SUMMA splitting is possible, we only investigate in this paper a simple way of extending the row-splitting coded computing to SUMMA splitting due to the increase in the communication cost when the entangled way is used. Our way is to only apply coding to each column block of the SUMMA matrix splitting, i.e., for $\Bbf_{ij},i=1,\ldots,\sqrt{k}$ and for each fixed $j$. More specifically, we partition all the $P$ workers evenly into $\sqrt{k}$ groups and partition $\x_t$ into $\sqrt{k}$ subvectors $\x_t^j$. Then, the $j$-th group is only responsible for computing $\w_t^j=\left[\begin{matrix}
  \Bbf_{1j}\x_t^j\\
  \Bbf_{2j}\x_t^j\\
  \vdots\\
  \Bbf_{\sqrt{k}j}\x_t^j\\
\end{matrix}\right]$. Since the workers in the $j$-th group is computing a matrix-vector multiplication with row-wise matrix splitting, one can directly apply the coded computing in the row-wise splitting case to compute an estimate of $\w_t^j$ using a $(P/\sqrt{k},\sqrt{k})$ sparse code. After that, the master node computes an estimate of $\Bbf\x_t = \sum_{j=1}^{\sqrt{k}}\w_t^j$ by adding the estimates of $\w_t^j,j=1,\ldots,\sqrt{k}$.

\section{Substitute Decoding for Coded Iterative Computing}\label{sec:decoding}

In Section~\ref{sec:coding_preliminaries1} and Section~\ref{sec:coding_preliminaries2}, we suggested a plausible tradeoff on the sparsity of $\Gbf$. If it is dense, storage cost is high. If it is sparse, $\Gbf_s^{(t)}$ may not be inverted to get the desired results $\Bbf \x_t+\y$. Surprisingly, we show that $\Gbf_s^{(t)}$ can actually be made sparse, while its noise-tolerance is maintained. We first examine the row-splitting case.

\subsection{Substitute Decoding Algorithm for LDGM codes for Row-wise Splitting}\label{sec:decoding_alg}

The key observation is that although $\Gbf_{s}^{(t)}$ may not have full column rank, we can get partial information of $\Bbf \x_t+\y$ from $\Gbf_s^{(t)}\mat{\Bbf \x_t+\y}$. Suppose that the SVD of $\Gbf_{s}^{(t)}$ is
\begin{equation}\label{eqn:G_svd}
(\Gbf_{s}^{(t)})_{(1-\epsilon)P\times k}=\Ubf_t\Dbf_t\Vbf_t^\top,
\end{equation}
where the matrix $\Vbf_t$ has orthonormal columns and has size $k\times \text{rank}(\Gbf_{s}^{(t)})$. By multiplying $\Lbf_t = \Dbf_t^{-1}\Ubf_t^\top $ to the partial coded results $\Gbf_s^{(t)}\mat{\Bbf \x_t+\y}$, the master node obtains
\begin{equation}\label{eqn:decoding_inverse}
\begin{split}
(\Dbf_t^{-1}\Ubf_t^\top)\Gbf_s^{(t)}\mat{\Bbf \x_t+\y}\overset{(a)}{=}\Vbf_t^\top\mat{\Bbf \x_t+\y},
\end{split}
\end{equation}
where ($a$) follows from \eqref{eqn:G_svd}. Then, the master node finds an orthonormal basis of the orthogonal complementary space of the column space of $\Vbf_t$, i.e., an orthonormal basis of $\Rcal^\bot(\Vbf_t)$ (where $\Rcal^\bot(\cdot)$ means the orthogonal complementary space), and forms the basis into a matrix
$\wt\Vbf_t$, such that the matrix $[\Vbf_t,\wt\Vbf_t]$ is an orthonormal one\footnote{If $\Gbf_{s}^{(t)}$ has full rank, $\Vbf_t$ is already a square orthonormal matrix and in this case $\wt\Vbf_t$ is the NULL matrix, because $\Rcal^\bot(\Vbf_t)$ is the trivial space $\{\zero\}$.} of size $k\times k$, i.e., $\Vbf_t$, $\wt\Vbf_t$ are orthogonal to each other, and
\begin{equation}\label{eqn:sum_vtilde_I}
\Vbf_t\Vbf_t^\top + \wt\Vbf_t\wt\Vbf_t^\top = \I_k.
\end{equation}
The master node uses $\Vbf_t^\top\mat{\Bbf \x_t+\y}$ obtained from \eqref{eqn:decoding_inverse} and the stored $\x_t$ as side information to obtain a good estimate of $\Bbf \x_t+\y$ to compute $\x_{t+1}$. In particular, $\x_{t+1}$ is
\begin{equation}\label{eqn:coded_pg}
\x_{t+1}=\vect{[\Vbf_t,\wt\Vbf_t]\cdot \left[\begin{matrix}
\Vbf_t^\top\mat{\Bbf \x_t+\y}\\
\wt\Vbf_t^\top\mat{\x_t}
\end{matrix}\right]}.
\end{equation}
This is equivalent to
\begin{equation}\label{eqn:coded_pg_2}
\x_{t+1}=\vect{\Vbf_t\Vbf_t^\top\mat{\Bbf \x_t+\y}+\wt\Vbf_t\wt\Vbf_t^\top\mat{\x_t}}.
\end{equation}
\begin{remark}\label{rem:scalar}
({\bf Intuition underlying substitute decoding}) We provide intuition by looking at a scalar-version as shown in the lower part of Fig.~\ref{fig:comparison}. In the scalar version, the vector length $b$ of each subvector of $\x_t$ satisfies $b=1$ and $\I_b=1$ and \eqref{eqn:coded_pg_2} becomes
\begin{equation}\label{eqn:coded_pg_scalar}
\x_{t+1}=\Vbf_t\Vbf_t^\top(\Bbf \x_t+\y)+\wt\Vbf_t\wt\Vbf_t^\top\x_t.
\end{equation}
Since $\Vbf_t$ and $\wt\Vbf_t$ have orthogonal columns and they are orthogonal to each other, $\Vbf_t\Vbf_t^\top$ and $\wt\Vbf_t\wt\Vbf_t^\top$ are two projection matrices onto the column spaces of $\Vbf_t$ and $\wt\Vbf_t$ respectively. Since $\Vbf_t$ is obtained from the SVD of $\Gbf_{s}^{(t)}$ (see equation \eqref{eqn:G_svd}), the projection $\Vbf_t\Vbf_t^\top$ is the projection to the row space of $\Gbf_{s}^{(t)}$. The intuition of substitute decoding is that even if we cannot get the exact result of $\Bbf \x_t+\y$ by inverting the sparse $\Gbf_s^{(t)}$, we can at least obtain the projection of $\Bbf \x_t+\y$ onto the row space of $\Gbf_{s}^{(t)}$. Then, for the remaining unknown part of $\Bbf \x_t+\y$, i.e., the projection of $\Bbf \x_t+\y$ onto the right null space of $\Gbf_{s}^{(t)}$, we use the projection $\wt\Vbf_t\wt\Vbf_t^\top\x_t$ of the side information $\x_t$ to substitute. This intuition is illustrated in Fig.~\ref{fig:sub_decoding}. We give an outline of the substitute decoding algorithm for the row-splitting case in Algorithm~\ref{alg:pg}.
\end{remark}
\begin{figure}
  \centering
  \includegraphics[scale=0.3]{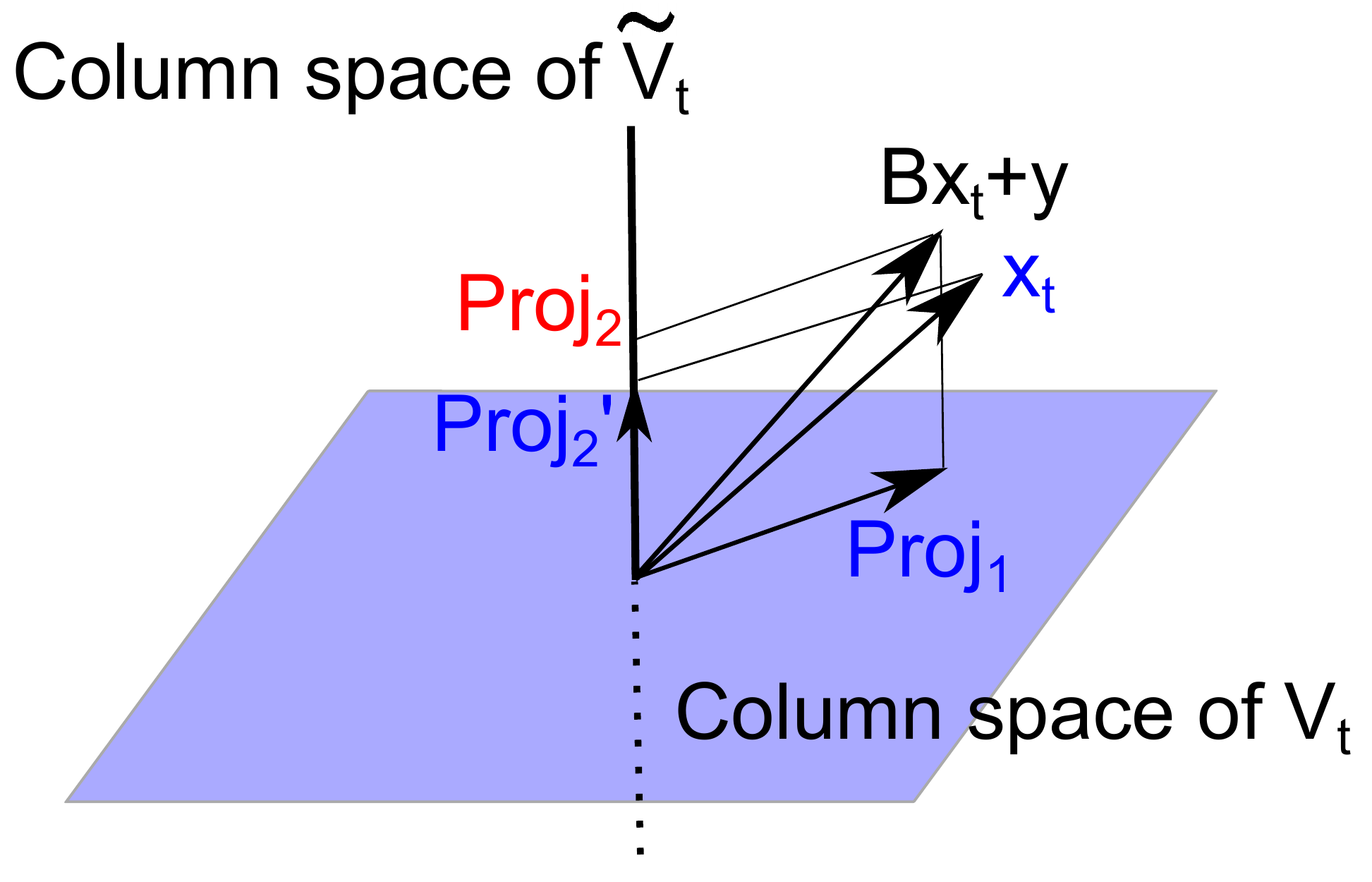}\\
  \caption{This is an illustration of substitute decoding where the known parts are colored \textcolor{blue}{blue} and the unknown parts are colored \textcolor{red}{red}. From $\Gbf_s^{(t)}(\Bbf\x_t+\y)$, we can get the projection of $\Bbf\x_t+\y$ onto the column space of $\Vbf_t$ (see Proj$_1$). For the unknown part Proj$_2$, we use the projection of $\x_t$ instead, which is Proj$_2$'.}\label{fig:sub_decoding}
\end{figure}

\begin{algorithm}[!h]
   \caption{Coded Power Iterations for Row-wise Splitting}\label{alg:pg}
\begin{algorithmic}
	\STATE {\bfseries Input:} Input $\y$, matrix $\Bbf$ and sparsity pattern $\Gbf$.
    \STATE {\bfseries Preprocessing:} Partition $\Bbf$ into row blocks and $\y$ into subvectors and store them distributedly as specified by the sparsity pattern matrix $\Gbf$. Generate a series of random generator matrices $\Gbf^{(t)},t=1,2,\ldots,T$.

   \STATE {\bfseries Master Node:} Send out $\x_t$ at each iteration and receive partial coded results. Compute $\x_{t+1}$ using substitute decoding \eqref{eqn:coded_pg_2}, where $\Vbf$, $\wt\Vbf$ are obtained from SVD \eqref{eqn:G_svd} and $\mat{\Bbf \x_t+\y}$ is computed using \eqref{eqn:decoding_inverse}.
	\STATE {\bfseries Workers:} Worker-$i$ computes $(\Gbf^{(t)})_\text{$i$-th row}\mat{\Bbf\x+\y}$.

\STATE {\bfseries Output:} The master node outputs $\x_T$.
\end{algorithmic}
\end{algorithm}

\subsection{Substitute Decoding for Column-wise Splitting}\label{sec:vertical}

In the section, we show how substitute decoding can be applied to the column-partition case described in Section~\ref{sec:col_part_preliminary}. From~\eqref{eqn:results_from_workers}, the results obtained from all the $P$ workers are
\begin{equation}\label{eqn:coded_results}
\Wbf_t=[\w_t^1,\w_t^2,\ldots,\w_t^P]=[\Bbf_1\x_t^1,\Bbf_2\x_t^2,\ldots,\Bbf_k\x_t^k]\cdot (\Gbf^{(t)})^\top.
\end{equation}
When the coded results from some workers are erased, the results at the mast node are $[\Bbf_1\x_t^1,\Bbf_2\x_t^2,\ldots,\Bbf_k\x_t^k]\cdot (\Gbf_s^{(t)})^\top$. Similar to the row-wise case, we compute the SVD of $\Gbf_s^{(t)}=\Ubf_t\Dbf_t\Vbf_t^\top$ and right-multiply the partial results by a decoding matrix $\Ubf_t\Dbf_t^{-1}\Vbf_t^\top$. Then, we obtain the projection $[\Bbf_1\x_t^1,\Bbf_2\x_t^2,\ldots,\Bbf_k\x_t^k]\Vbf_t\Vbf_t^\top$. Define $\u_t^j=\Bbf_j\x_t^j,j=1,\ldots,k$. The master node always maintains an estimate of $\u_t^j$ which we denote by $\widehat{\u}_t^j$. Then, the substitute decoding step is defined by
\begin{equation}\label{eqn:substitute_decoding_vertical}
\begin{split}
[\widehat{\u}_t^1,\ldots,\widehat{\u}_t^k]=&[\Bbf_1\x_t^1,\Bbf_2\x_t^2,\ldots,\Bbf_k\x_t^k]\Vbf_t\Vbf_t^\top\\
&+[\widehat{\u}_{t-1}^1,\ldots,\widehat{\u}_{t-1}^k]\wt\Vbf_t\wt\Vbf_t^\top.
\end{split}
\end{equation}
Finally, the master node uses $\sum_{j=1}^k \widehat{\u}_t^j$ as the estimate of $\Bbf\x_t$ and updates $\x_{t+1}$ as follows
\begin{equation}\label{eqn:x_update_vertical}
\x_{t+1} = [\widehat{\u}_t^1,\ldots,\widehat{\u}_t^k]\cdot \one_k + \y.
\end{equation}

A potential drawback of the decoding step \eqref{eqn:substitute_decoding_vertical} when compared to the decoding step \eqref{eqn:coded_pg_2} in the row-wise splitting case is that the decoding complexity is higher due to matrix-matrix multiplication on (a submatrix of) the $N\times P$ coded partial results matrix $\Wbf_t$. One way to reduce the time cost is to compute
\begin{equation}\label{eqn:substitute_decoding_vertical_real}
\begin{split}
  \x_{t+1} =& [\widehat{\u}_t^1,\ldots,\widehat{\u}_t^k]\cdot \one_k + \y\\
  \overset{(a)}{=}& [\Bbf_1\x_t^1,\Bbf_2\x_t^2,\ldots,\Bbf_k\x_t^k](\Vbf_t\Vbf_t^\top\one_k)\\
&+[\widehat{\u}_{t-1}^1,\ldots,\widehat{\u}_{t-1}^k](\wt\Vbf_t\wt\Vbf_t^\top\one_k)+\y\\
  \overset{(b)}{=}& [\w_t^1,\w_t^2,\ldots,\w_t^P](\Ubf_t\Dbf_t^{-1}\Vbf_t^\top\one_k)\\
&+[\widehat{\u}_{t-1}^1,\ldots,\widehat{\u}_{t-1}^k](\wt\Vbf_t\wt\Vbf_t^\top\one_k)+\y,
\end{split}
\end{equation}
where step (a) is from \eqref{eqn:substitute_decoding_vertical} and step (b) is from right-multiplying $\Ubf\Dbf_t^{-1}\Vbf_t^\top$ to \eqref{eqn:coded_results}. Here, computing $\Ubf_t\Dbf_t^{-1}\Vbf_t^\top\one_k$ and $\wt\Vbf_t\wt\Vbf_t^\top\one_k$ have complexity $O(Pk)$, and the remaining part is still matrix-vector multiplication which has complexity $O(PN)$, instead of the matrix-matrix multiplication in \eqref{eqn:substitute_decoding_vertical} which has complexity $O(PkN)$. After computing $\x_{t+1}$, the master node starts transmitting $\x_{t+1}$ to the workers. In the meantime, the master node starts a new thread to compute \eqref{eqn:substitute_decoding_vertical} and updates the estimates $[\widehat{\u}_t^1,\ldots,\widehat{\u}_t^k]$. We provide an outline of substitute decoding for the column-wise splitting case in Algorithm~\ref{alg:pg_col}.

\begin{algorithm}[!h]
   \caption{Coded Power Iterations for Column-wise Splitting}\label{alg:pg_col}
\begin{algorithmic}
	\STATE {\bfseries Input:} Input $\y$, matrix $\Bbf$ and sparsity pattern $\Gbf$.
    \STATE {\bfseries Preprocessing:} Partition $\Bbf$ into column blocks and store them distributedly as specified by the sparsity pattern $\Gbf$. Generate a series of random generator matrices $\Gbf^{(t)},t=1,2,\ldots,T$.

   \STATE {\bfseries Master Node:} At each iteration, partition $\x_t$ into $k$ subvectors $\x_t^j,j=1,\ldots,k$ and transmits $\x_t^j$ to each worker $i$ such that $\Gbf_{ij}=1$. Then, receive part of all the coded results $[\w_t^1,\w_t^2,\ldots,\w_t^P]$. Conduct substitute decoding \eqref{eqn:substitute_decoding_vertical_real} and get $\x_{t+1}$. While transmitting $\x_{t+1}$ to the other workers, update the estimates $[\widehat{\u}_t^1,\ldots,\widehat{\u}_t^k]$ using \eqref{eqn:substitute_decoding_vertical}.
	\STATE {\bfseries Workers:} The $i$-th worker computes $\w_t^i$ in \eqref{eqn:worker_vertical}.

\STATE {\bfseries Output:} The master node outputs $\x_T$.
\end{algorithmic}
\end{algorithm}

\subsection{Substitute Decoding for SUMMA-type splitting}\label{sec:SUMMA_decoding}

As we have discussed in Section~\ref{sec:SUMMA}, we extend the substitute decoding algorithm in the case of row-wise splitting to SUMMA splitting. The way is to apply separate substitute decoding to the coded computation of $\w_t^j=\left[\begin{matrix}
  \Bbf_{1j}\x_t^j\\
  \Bbf_{2j}\x_t^j\\
  \vdots\\
  \Bbf_{\sqrt{k}j}\x_t^j\\
\end{matrix}\right]$ in each group of $P/\sqrt{k}$ workers for all $\sqrt{k}$ groups. Therefore, the algorithm of substitute decoding in this case is almost exactly the same as in the row-splitting case. The only difference is that when the master node waits for $k$ out of $P$ workers to finish the computation, the number of workers that responds in each group may be different from each other, because the number of erasures in each group may not be the same. Therefore, it would be useful to have a decoding algorithm that can work flexibly in terms of the number of workers that responds. As we have shown in Section~\ref{sec:decoding_alg} and Section~\ref{sec:vertical}, the substitute decoding algorithm indeed has this advantage: it does not need to specify the number of workers that responds in order to decode. Therefore, we can easily transplant the row-wise splitting case to SUMMA splitting without concerning about different number of responses in different groups.

\subsection{Cost Analysis of Substitute Decoding}\label{sec:complexity}
\subsubsection{Cost Analysis in Algorithm~\ref{alg:pg}}

First, let us analyze the communication cost. For each worker, the communication cost at each iteration comes from the transmission of $\x_t$ of length $N$ and the transmission of the result of length $b=N/k$. So the total number of communicated floating point numbers is $N(1+1/k)$ which is linear in $N$.

For the computation cost at the master node, the SVD has complexity $O(kP^2)$ which is negligible. The computation of \eqref{eqn:decoding_inverse} and the substitute decoding given in \eqref{eqn:coded_pg} together have complexity $O(k^2 b)=O(kN)$ which is linear in $N$. The vectorization and matricization steps have a negligible cost.

The $i$-th worker computes a sub-vector of the entire $\Bbf\x$. Denote by $E$ the number of non-zeros in $\Bbf$. Suppose the sparse generator matrix $\Gbf$ has $d$ ones in each row. Then, the complexity at each worker is $O((dE)/k)$. This complexity can be superlinear in $N$ for dense graphs, but usually, the average degree of a graph is a large constant. This means the computation cost at each worker is also linear in $N$, but with a large constant.

\subsubsection{Cost Analysis in Algorithm~\ref{alg:pg_col}}

First, we analyze the communication cost. For each worker, the communication complexity comes from receiving a constant number of $\x_t^j$ of length $N/k$ and the transmission of the result of length $N$. The number of $\x_t^j$ is $d$ where $d$ is the number of ones in each row of the matrix. So the total number of communicated floating point numbers is $N(1+d/k)$ which is linear in $N$.

For the computational cost at the master node, the SVD is like the row-wise splitting case and has negligible cost. As we have analyzed in Section~\ref{sec:vertical}, the substitute decoding step \eqref{eqn:substitute_decoding_vertical_real} has complexity $O(PN)$, and the computation of \eqref{eqn:substitute_decoding_vertical} has complexity $O(PkN)$ but can be done simultaneously with the communication step. The computational complexity at each worker is the same as in the row-splitting case and is $O((dE)/k)$.

\begin{remark}
Since the communication complexity and the computational complexity are both linear in $N$, constant factors matter in determining which cost is dominant. When the code size is small, the communication time can dominate even when the average degree of the graph is much larger than 1, because, in practice, communication time can dominate computing time even if computing is in scaling sense more expensive, simply due to the fact that communication is slower than computation in current technologies \cite{li2018fundamental,tandon2016gradient}. In our simulations, we focus on the communication cost and analyze the cost mostly from this perspective. However, we still cannot use a dense code for encoding due to increased storage cost. We also cannot fit the entire graph into the main memory of a single machine because the average degree can be a large constant.
\end{remark}

\subsection{Convergence Analysis of the Coded Power Iterations Using Substitute Decoding}

\subsubsection{Analysis of Algorithm \ref{alg:pg}}

Denote by $\x^*$ the true solution of $\x=\Bbf\x+\y$. Then,
\begin{equation}\label{eqn:true_iter}
\begin{split}
\x^* \overset{(a)}{=}& \vect{\Vbf_t\Vbf_t^\top\mat{\x^*}} + \vect{\wt\Vbf_t\wt\Vbf_t^\top\mat{\x^*}}\\
=&  \vect{\Vbf_t\Vbf_t^\top\mat{\Bbf \x^*+\y}+\wt\Vbf_t\wt\Vbf_t^\top\mat{\x^*}},
\end{split}
\end{equation}
where ($a$) holds because of \eqref{eqn:sum_vtilde_I}. Defining the remaining error as $\e_t=\x_t-\x^*$ and subtracting \eqref{eqn:true_iter} from \eqref{eqn:coded_pg_2}, we have
\begin{equation}\label{eqn:err_iter}
\e_{t+1}=\vect{\Vbf_t\Vbf_t^\top\mat{\Bbf \e_t}}+\vect{\wt\Vbf_t\wt\Vbf_t^\top\mat{\e_t}}.
\end{equation}

\begin{remark}(Why substitute decoding suppresses error) Before presenting formal proofs, we show the underlying intuition on why the substitute decoding \eqref{eqn:coded_pg} approximates the noiseless power iteration $\x_{t+1}=\Bbf\x_t+\y$ well. Again, we look at the scalar version, i.e., $\I_b=1$. In this case, \eqref{eqn:err_iter} becomes
\begin{equation}\label{eqn:err_suppress}
\e_{t+1}=\Vbf_t\Vbf_t^\top\Bbf \e_t+\wt\Vbf_t\wt\Vbf_t^\top\e_t.
\end{equation}
Ideally, we want $\e_{t+1}=\Bbf\e_t$ since $\Bbf$ is a contraction matrix (recall that $\rho(\Bbf)<1$). Due to noise, we can only realize this contraction in the column space of $\Vbf_t$, which is the first term $\Vbf_t\Vbf_t^\top\Bbf \e_t$. Although the partial generator matrix $\Gbf_s^{(t)}$ may not have full rank due to being sparse (i.e., $\text{dim}(\Rcal(\Vbf_t))<k$), $\Gbf_s^{(t)}$ can be close to full rank. This means that the column space $\Rcal(\wt\Vbf_t)$ can have low dimension. Therefore, for the second term in \eqref{eqn:err_suppress}, the projection matrix $\wt\Vbf_t\wt\Vbf_t^\top$ suppresses the larger error $\e_t$ (compared to $\Bbf \e_t$) by projecting it onto the low-dimensional space $\Rcal(\wt\Vbf_t)$. In Theorem~\ref{thm:converge}, we will show $\wt\Vbf_t\wt\Vbf_t^\top$ reduces $\Ep[\twonorm{\e_t}^2]$ by a small multiple factor that decreases to 0 linearly as $\text{rank}(\Gbf_s^{(t)})$ increases.
\end{remark}

\begin{definition}\label{def:cyclic} (Combined cyclic sparsity pattern)
The sparsity pattern matrix satisfies $\Gbf=\left[\begin{matrix}
\Sbf_1\\
\Sbf_2
\end{matrix}\right]$,
where $\Sbf_1$ and $\Sbf_2$ are both $k\times k$ square cyclic matrices with $d$ non-zeros in each row.
\end{definition}

\begin{assumption} (Random failures) \label{ass:noise}
At each iteration, a random subset of the workers fail to compute the result due to either stragglers or erasure-type errors. Failure events are independent across all iterations.
\end{assumption}

\begin{theorem}(Convergence Rate of Algorithm~\ref{alg:pg})\label{thm:converge}
If the sparsity pattern $\Gbf$ in Definition~\ref{def:cyclic} is used and Assumption~\ref{ass:noise} holds, the remaining error $\e_t=\x_t-\x^*$ of Algorithm~\ref{alg:pg} satisfies
\begin{equation}
\Ep[\twonorm{\e_{t+1}}^2]=(1-\delta_t) \Ep[\twonorm{\Bbf\e_{t}}^2]+\delta_t\Ep[\twonorm{\e_{t}}^2],
\end{equation}
where
\begin{equation}
\delta_t=1-\frac{\Ep[\text{rank}(\Gbf_s^{(t)})]}{k}.
\end{equation}
\end{theorem}
The proof is in Section~\ref{sec:proof}. From Theorem~\ref{thm:converge}, we can simply upper-bound $\Ep[\twonorm{\Bbf\e_{t}}^2]$ by $\twonorm{\Bbf}_2^2\Ep[\twonorm{\e_{t}}^2]$ and hence
\begin{equation}
\Ep[\twonorm{\e_{t+1}}^2]\le [(1-\delta_t)\twonorm{\Bbf}_2^2+\delta_t]\cdot\Ep[\twonorm{\e_t}^2].
\end{equation}
This means that when $\delta_t$ is close to 0, i.e., when $\Gbf_s^{(t)}$ is close to full rank, $\Ep[\twonorm{\e_{t}}^2]$ converges to $0$ with rate close to $\twonorm{\Bbf}_2^{2t}$. In Table \ref{table:ds} in Section~\ref{sec:simulation}, we show how $\delta_t$ changes with the degree $d$ in Definition~\ref{def:cyclic}. Notice that the noiseless power iterations converge with rate $(\rho(\Bbf))^{2t}$. For the PageRank problem, $\Bbf=(1-c)\Abf$ and $\Abf$ is the column-normalized adjacency matrix. We show in Lemma~\ref{lmm:norm_close} that for Erd\"os-R\'enyi model $G(N,p)$, $\text{Pr}\left(\twonorm{\Abf}>\sqrt{\frac{1+\epsilon}{1-\epsilon}}\rho(\Abf)\right)<3Ne^{-\epsilon^2Np/8}$. This means that with high probability $\twonorm{\Abf}_2\approx\rho(\Abf)$ and hence $\twonorm{\Bbf}_2\approx\rho(\Bbf)$, and the convergence rate of coded power iteration and that of noiseless power iteration are close. Here, in $G(N,p)$, it suffices for $p$ to be $\Omega(\log N/(N\epsilon^2))$ for $3Ne^{-\epsilon^2Np/8}$ to be small.

\subsubsection{Analysis of Algorithm~\ref{alg:pg_col}}

Recall that $\x^*$ is the true solution of $\x=\Bbf\x+\y$. Denote by $\x^{j*}$ the $j$-th subvector of $\x^*$ after partitioning $\x^*$ into $k$ subvectors. Suppose $\u^{j*}=\Bbf_j\x^{j*}$. Then, the intermediate estimates $[\widehat{\u}_t^1,\ldots,\widehat{\u}_t^k]$ in \eqref{eqn:substitute_decoding_vertical} can also be viewed as estimates of $\u^{j*}=\Bbf_j\x^{j*}, j=1, \ldots, k$. Define $\e_t^j=\widehat{\u}_t^j-\u^{j*}$. Notice that the optimal solution $\x^*$ satisfies
\begin{equation}\label{eqn:error_analysis_1}
\begin{split}
\x^*=\Bbf\x^*+\y
=\sum_{j=1}^k \Bbf_j\x^{j*} +\y
=\sum_{j=1}^k \u^{j*} +\y.
\end{split}
\end{equation}
Subtracting \eqref{eqn:error_analysis_1} from \eqref{eqn:x_update_vertical}, we have
\begin{equation}\label{eqn:e_sum}
  \e_{t+1} = \sum_{j=1}^k \e_t^j,
\end{equation}
where $\e_{t+1}=\x_{t+1}-\x^*$. Therefore, if we can prove the convergence of $\e_t^j$ to 0, we can simultaneously establish the convergence of $\e_{t}$, and hence prove the convergence of $\x_t$ to $\x^*$. Define
\begin{equation}
E_t=\left[\begin{matrix}
  \e_t^1\\
  \vdots\\
  \e_t^k
\end{matrix}\right].
\end{equation}
Then, the following theorem state the convergence rate of $E_t$.

\begin{theorem}(Convergence Rate of Algorithm~\ref{alg:pg_col})\label{thm:converge_col}
If the sparsity pattern $\Gbf$ in Definition~\ref{def:cyclic} is used and Assumption~\ref{ass:noise} holds, the remaining error $E_t$ of Algorithm~\ref{alg:pg_col} satisfies
\begin{equation}\label{eqn:converge_rate_col}
\Ep[\twonorm{E_{t+1}}^2]\le [(1-\delta_t)\twonorm{\Bbf}_\text{col}^2+\delta_t] \Ep[\twonorm{E_t}^2],
\end{equation}
where
\begin{equation}
\delta_t=1-\frac{\Ep[\text{rank}(\Gbf_s^{(t)})]}{k},
\end{equation}
and $\twonorm{\Bbf}_\text{col}=\sqrt{k}\max_{j}\{\twonorm{\Bbf_j}_2\}$ and $\Bbf_j$ are the column blocks of $\Bbf$.
\end{theorem}
From Theorem~\ref{thm:converge_col}, we can see that apart from the $\delta_t$ factor again, we have defined a new norm $\twonorm{\Bbf}_\text{col}$, which is the maximum of the induced 2-norms of the column blocks $\Bbf_1,\ldots \Bbf_k$. We show in Lemma \ref{lmm:norm_close_col} that for Erd\"os-R\'enyi model $G(N,p)$, $\text{Pr}\left(\twonorm{\Abf}_\text{col}>\sqrt{\left(1+\frac{k}{N}\right)\frac{1+\epsilon}{1-\epsilon}}\rho(\Abf)\right)\\
  <3Ne^{-\epsilon^2Np/8}+3kNe^{-\epsilon^2Np/(8k)}$. This means that with high probability $\twonorm{\Abf}_\text{col}\approx\rho(\Abf)$ and hence $\twonorm{\Bbf}_\text{col}\approx\rho(\Bbf)$, and the convergence rates of coded power iterations and noiseless power iterations are close.

\section{Applications of Substitute Decoding}\label{sec:extension}

Substitute decoding can be applied to many iterative computing problems. In this section, we show three applications, namely computing multiple leading eigenvectors, computing multiple leading singular vectors and gradient descent. The first two applications are widely used in sparse matrix problems such as spectral clustering \cite{von2007tutorial}, principal component analysis and anomaly detection \cite{prakash2010eigenspokes}.

\subsection{Computing Multiple Eigenvectors using Coded Orthogonal Iterations}\label{sec:spectral_clustering}

\subsubsection{Background on the Orthogonal-Iteration Method}
When the input vector $\y$ is zero, the power-iteration method in \eqref{eqn:power_iteration} is equivalent to computing the principal eigenvector of $\Bbf$. However, we may be interested in more than one eigenvectors. For example, in spectral clustering or spectral embedding, instead of computing a single eigenvector, one often computes the first $r$ eigenvectors of the (normalized) graph Laplacian matrix $\Lbf=\I-\Dbf^{-1/2}\Abf\Dbf^{-1/2}$ \cite{ng2002spectral} and use these eigenvectors as the coordinates of the $r$-dimensional Euclidean embedding of the nodes in the graph. Here, $\Abf$ is the (symmetric) graph adjacency matrix and $\Dbf$ is the diagonal degree matrix\footnote{Notice that in spectral clustering, we need the first $r$ eigenvectors corresponding to the $r$ smallest eigenvalues instead of the largest eigenvalues. However, it is well-known that (1) all eigenvalues of $\Lbf=\I-\Dbf^{-1/2}\Abf\Dbf^{-1/2}$ are within the range $[0,2]$, and (2) since $\Lbf$ is symmetric, all the eigenvectors of $\Lbf$ form an orthonormal basis and are also the eigenvectors of $\I$. Therefore, the $r$ eigenvectors corresponding to the $r$ smallest eigenvalues of $\Lbf$ are also the $r$ eigenvectors corresponding to the $r$ largest eigenvalues of $2\I-\Lbf=\I+\Dbf^{-1/2}\Abf\Dbf^{-1/2}$.}.

The power-iteration method can be generalized to the \emph{orthogonal-iteration} method to compute the first $r$ eigenvectors \cite{bauer1957verfahren}. To compute the first $r$ eigenvectors of an $N\times N$ matrix $\Bbf$, we initialize an $N\times r$ random matrix $\Xbf_0$ and iterate the following until convergence:
\begin{itemize}
    \item Compute
    \begin{equation}\label{eqn:QR_before}
        \Zbf_t=\Bbf \Xbf_t.
    \end{equation}
    \item Factorize
    \begin{equation}\label{eqn:QR}
       \Qbf_t\Rbf_t=\Zbf_t,
    \end{equation}
    using QR-decomposition.
    \item Set $\Xbf_{t+1}=\Qbf_t$.
\end{itemize}
Usually the $N\times r$ matrix $\Xbf_t$ is tall and thin, because the number of required eigenvectors $r$ is much less than $N$.

The orthogonal-iteration method in \eqref{eqn:QR_before} and \eqref{eqn:QR} is the prototype of many large-scale eigendecomposition methods \cite{rutishauser1969computational,stewart1976simultaneous}\cite[Section 7.3.2]{golub2012matrix}\cite{berry2006parallel,halko2011finding,kempe2008decentralized,ji2016apache}. For example, to accelerate the convergence of the orthogonal-iteration method for a symmetric $\Bbf$, one can apply the following procedure as suggested in \cite{rutishauser1969computational,rutishauser1970simultaneous,stewart1976simultaneous,berry2006parallel} after obtaining $\Qbf_t$ and $\Rbf_t$ in \eqref{eqn:QR}:
\begin{itemize}
    \item Compute
        \begin{equation}\label{eqn:acc1}
            \Dbf_t=\Rbf_t\Rbf_t^\top.
        \end{equation}
    \item Compute the eigendecomposition
        \begin{equation}\label{eqn:acc2}
            \Sbf_t\boldsymbol{\Lambda}_t\Sbf_t^\top=\Dbf_t.
        \end{equation}
    \item Compute the modified eigenvectors
        \begin{equation}\label{eqn:acc3}
            \Xbf_{t+1}=\Qbf_t\Sbf_t.
        \end{equation}
\end{itemize}
Another method is to apply QR-decomposition \eqref{eqn:QR} only once after computing \eqref{eqn:QR_before} several times in each iteration, which has the effect of making the spectrum of $\Bbf$ more skewed and making the convergence faster \cite{halko2011finding}.

\subsubsection{Coded Orthogonal Iterations for Computing Multiple Eigenvectors}
We show how to implement a coded version of the orthogonal-iteration method to compute the first $r$ eigenvectors of an $N\times N$ matrix $\Bbf$. The number $r$ is much less than the size of $\Bbf$, so the QR-decomposition \eqref{eqn:QR} can be directly computed at the master node. The only computation at the workers is the matrix-matrix multiplication $\Zbf_t=\Bbf \Xbf_t$ in \eqref{eqn:QR_before}.

The procedures of coded orthogonal iterations are outlined in Algorithm~\ref{alg:oi}. Similar to the computation of $\Bbf \x_t+\y$ in Algorithm~\ref{alg:pg_col}, we partition the matrix $\Bbf$ into column blocks $[\Bbf_1,\ldots,\Bbf_k]$ and distribute them to the workers as specified by the sparsity pattern matrix $\Gbf$. At each iteration, the master node breaks the $\Xbf_t$ into $k$ submatrices $\Xbf_t^j,j=1,\ldots,k$ and each worker computes a linear combination of $\Wbf_{j,t}:=\Bbf_j\Xbf_t^j,j=1,\ldots,k$. The $i$-th worker computes the linear combination $\sum_{j=1}^k g_{ij}^t \Wbf_{j,t}$. The collected results at the master node, if no noise is present, can be compactly written as $\bar{\Wbf}_t(\Gbf^{(t)})^\top$, where $\bar{\Wbf}_t=[\Wbf_{1,t},\Wbf_{2,t},\ldots,\Wbf_{k,t}]$. Notice that a rigorous way to write $\bar{\Wbf}_t(\Gbf^{(t)})^\top$ is to change $\Gbf^{(t)}$ into $\Gbf^{(t)}\otimes \I$ to match the matrix sizes. However, to avoid cumbersome notation and provide a clean presentation of the main idea, we view $\Wbf_{j,t},j=1,\ldots,k$ as symbols and still uses $\Gbf^{(t)}$. Again, similar to Algorithm~\ref{alg:pg_col}, the master node maintains an estimate of all the $\Wbf_{j,t}$, which can be written as $\widehat{\Wbf}_t = [\widehat{\Wbf}_{1,t},\ldots,\widehat{\Wbf}_{k,t}]$. In the presence of erasure noise, the master node can combine the  partial coded results $\bar{\Wbf}_t(\Gbf_s^{(t)})^\top$ and the previous results $\widehat{\Wbf}_{t-1}$ to obtain the current estimate $\widehat{\Wbf}_t$.

In order to accelerate the convergence, we further apply the procedures from \eqref{eqn:acc1} to \eqref{eqn:acc3}. This modification has to be applied carefully to the coded computing because the modified eigenvectors in \eqref{eqn:acc3} may not have the same order as in the previous iteration. Therefore, the naive combination of the results from the past and the current iteration in substitute decoding can be affected by the order of the eigenvectors. To address this problem, we notice that \eqref{eqn:QR}-\eqref{eqn:acc3} are equivalent to the following:
\begin{itemize}
  \item Factorize
    \begin{equation}\label{eqn:acc21}
       \Qbf_t\Rbf_t=\Zbf_t,
    \end{equation}
  \item Compute the SVD
    \begin{equation}\label{eqn:acc22}
        \Rbf_t = \Sbf_t\boldsymbol{\Lambda}_t^{\frac{1}{2}} \widetilde{\Sbf}_t^\top,
    \end{equation}
  \item Compute the modified eigenvectors
    \begin{equation}\label{eqn:acc23}
        \Xbf_{t+1}=\Qbf_t\Sbf_t = \Zbf_t\Rbf_t^{-1}\Sbf_t = \Zbf_t \widetilde{\Sbf}_t\boldsymbol{\Lambda}_t^{-\frac{1}{2}}.
    \end{equation}
\end{itemize}
Thus, we can see that the modified QR-steps in \eqref{eqn:acc21}-\eqref{eqn:acc23} essentially right-multiplies a matrix $\widetilde{\Sbf}_t\boldsymbol{\Lambda}_t^{-\frac{1}{2}}$ to $\Zbf_t$, in which the matrix $\widetilde{\Sbf}_t$ has the function of reordering the eigenvectors, because $\boldsymbol{\Lambda}_t^{-\frac{1}{2}}$ is only a diagonal matrix. Therefore, at each iteration, we apply the same reordering to $\widehat{\Wbf}_t$ and obtain
\begin{equation}\label{eqn:eigen_rotate}
\widehat{\Wbf}_t^\text{rotate} = [\widehat{\Wbf}_{1,t}\widetilde{\Sbf}_t,\ldots,\widehat{\Wbf}_{k,t}\widetilde{\Sbf}_t],
\end{equation}
and uses $\widehat{\Wbf}_{t-1}^\text{rotate}$ from last iteration for substitute decoding at the $t$-th iteration
\begin{equation}\label{eqn:substitute_decoding_sc}
\widehat{\Wbf}_{t}=\bar{\Wbf}_t\Vbf_t\Vbf_t^\top+\widehat{\Wbf}_{t-1}^\text{rotate}\widetilde{\Vbf}_t\widetilde{\Vbf}_t^\top,
\end{equation}
where $\bar{\Wbf}_t\Vbf_t\Vbf_t^\top$ are computed from $\bar{\Wbf}_t(\Gbf^{(t)})^\top$ using the SVD on $\Gbf^{(t)}$. The summation of the symbols in $\widehat{\Wbf}_{t}$ is the estimate of $\Zbf_t$:
\begin{equation}
\Zbf_t = \sum_{j=1}^k \widehat{\Wbf}_{j,t},
\end{equation}
Then, instead of \eqref{eqn:QR}-\eqref{eqn:acc3}, the master node performs the equivalent steps \eqref{eqn:acc21}-\eqref{eqn:acc23} to obtain $\Xbf_{t+1}$. In order to reduce the decoding time, the master can apply the same method as in Section \ref{sec:vertical} and computes $\Zbf_t$ directly using
\begin{equation}\label{eqn:substitute_decoding_sc_2}
\Zbf_t=\bar{\Wbf}_t(\Gbf_s^{(t)})^\top[\Ubf_t\Dbf_t^{-1}\Vbf_t^\top\one_k]+\widehat{\Wbf}_{t-1}^\text{rotate}[\widetilde{\Vbf}_t\widetilde{\Vbf}_t^\top\one_k],
\end{equation}
and updates \eqref{eqn:substitute_decoding_sc} and \eqref{eqn:eigen_rotate} while communicating with the workers.

\begin{algorithm}[!h]
   \caption{Computing $r$ Eigenvectors Using Coded Orthogonal Iterations with Column Splitting}\label{alg:oi}
\begin{algorithmic}
    \STATE {\bfseries Input:} Matrix $\Bbf$ and sparsity pattern $\Gbf$.
    \STATE {\bfseries Preprocessing:} Partition $\Bbf$ into column blocks and store them distributedly as specified by the sparsity pattern $\Gbf$. Generate a series of random generator matrices $\Gbf^{(t)},t=1,2,\ldots,T$.

   \STATE {\bfseries Master Node:} At each iteration, partition $\Xbf_t$ into $k$ submatrices $\Xbf_t^j,j=1,\ldots,k$ and transmits $\Xbf_t^j$ to all worker $i$ such that $\Gbf_{ij}=1$. Then, receive partial coded results of $\Bbf \Xbf_t$ and conduct substitute decoding \eqref{eqn:substitute_decoding_sc_2} and get $\Zbf_t$.

   Perform the QR-decomposition steps \eqref{eqn:acc21}-\eqref{eqn:acc23} and set $\Xbf_{t+1}=\Qbf_t$.

   While transmitting $\Xbf_{t+1}$ to the other workers, the master node updates the estimates \eqref{eqn:substitute_decoding_sc} and \eqref{eqn:eigen_rotate}.

	\STATE {\bfseries Workers:} The $i$-th worker computes $\sum_{j=1}^k g_{ij}^t \Wbf_{j,t}$.

\STATE {\bfseries Output:} The master node outputs $\Xbf_T$.
\end{algorithmic}
\end{algorithm}

\subsection{Computing Multiple Singular Vectors using Coded Orthogonal Iterations}\label{sec:pca}

In the previous section, we discussed the application of computing $r>1$ eigenvectors using coded orthogonal iterations. Here, we discuss the computation of singular vectors or the truncated SVD. In this section, we focus on the case when the data matrix $\Bbf$ of size $n\times N$ is large and sparse, and directly computing the $N\times N$ matrix $\Bbf^\top \Bbf$ is not feasible because storing an $N\times N$ dense matrix $\Bbf^\top \Bbf$ is impractical even for distributed storage. However, we can still partition the sparse data matrix $\Bbf$ into blocks and store them distributedly.

One application of SVD is principal component analysis (PCA). In PCA, the $j$-th principal component of the $n\times N$ data matrix $\Bbf$ is defined as the $j$-th column of $\Bbf\Xbf$, where $\Xbf$ is formed by the eigenvectors of the covariance matrix $\Bbf^\top \Bbf$, or the right-singular vectors of $\Bbf$. In most dimension-reduction applications, we only need to compute $r$ leading singular vectors of $\Bbf$, where $r\ll N$. The PCA can help reduce the dimension of the data substantially, in order to obtain reduced generalization error for classification tasks. Another application of computing SVD is to obtain the spectral embedding of a large and sparse matrix in order to detect anomalies \cite{prakash2010eigenspokes}, which we will discuss further in Section~\ref{sec:pca_simu}.

\subsubsection{Noiseless computation}\label{sec:pca_noiseless}
The right-singular vectors of $\Bbf$ are the eigenvectors of $\Bbf^\top \Bbf$, so the methods discussed in this section are essentially the same with Section~\ref{sec:spectral_clustering}. However, the computation of $\Bbf^\top \Bbf$ is infeasible, i.e., we cannot compute the matrix-matrix multiplication step $\Zbf_t=\Bbf^\top\Bbf \Xbf_t$ as in \eqref{eqn:QR_before} by computing $\Bbf^\top \Bbf$ and dividing it into column blocks. If the computation is noiseless, what we can do is to compute the matrix-matrix multiplication step $\Zbf_t=\Bbf^\top\Bbf \Xbf_t$ distributedly in the following way:
\begin{equation}\label{eqn:local_w}
\text{At workers: } \Wbf_{i,t}=\Bbf_i^\top\Bbf_i \Xbf_t,
\end{equation}
\begin{equation}\label{eqn:sum_z}
\text{At the master node: } \Zbf_t=\sum_{i=1}^P \Wbf_{i,t},
\end{equation}
where $\Bbf_i$ is the $i$-th row block in $\Bbf$ and $\Bbf_i$ is stored locally at the $i$-th worker. Notice that for uncoded computation the number of splits $k$ equals the number of workers $P$. At the $t$-th iteration, the master node broadcasts $\Xbf_t$ to all the workers and receives the result $\Wbf_{i,t}=\Bbf_i^\top\Bbf_i \Xbf_t$ from the $i$-th worker. If the computation is noisy and some of the partial results $\Wbf_{i,t}$ are not available, the master node simply uses the previous computational results $\Wbf_{i,t-1}$. After computing $\Zbf_t$, the master node computes the QR-decomposition in \eqref{eqn:acc21}-\eqref{eqn:acc23} to obtain $\Xbf_{t+1}$. Notice that even in the uncoded computation, since the acceleration procedure can rotate the order of eigenvectors, the master node needs to apply the rotation step \eqref{eqn:eigen_rotate}.

\subsubsection{Coded computation}

The method in this part is essentially the same as Algorithm~\ref{alg:oi} so we will skip the details. The only difference is that the partial result is defined as $\Wbf_{j,t}=\Bbf_j^\top\Bbf_j \Xbf_t$ and the master needs to broadcast the entire $\Xbf_t$ to all the workers.

\begin{algorithm}[!h]
   \caption{Computing $r$ Singular Vectors Using Coded Orthogonal Iterations with Column Splitting}\label{alg:pca}
\begin{algorithmic}
    \STATE {\bfseries Input:} Matrix $\Bbf$ and sparsity pattern $\Gbf$.
    \STATE {\bfseries Preprocessing:} Partition $\Bbf$ into row blocks and store them distributedly as specified by the sparsity pattern $\Gbf$. Generate a series of random generator matrices $\Gbf^{(t)},t=1,2,\ldots,T$.

   \STATE {\bfseries Master Node:} At each iteration, broadcast $\Xbf_t$ to all workers. Then, receive partial coded results of $\Bbf^\top\Bbf \Xbf_t$ and conduct substitute decoding \eqref{eqn:substitute_decoding_sc_2} and get $\Zbf_t$.

   Perform the QR-decomposition steps \eqref{eqn:acc21}-\eqref{eqn:acc23} and set $\Xbf_{t+1}=\Qbf_t$.

   While transmitting $\Xbf_{t+1}$ to the workers, the master node updates the estimates \eqref{eqn:substitute_decoding_sc} and \eqref{eqn:eigen_rotate}.

	\STATE {\bfseries Workers:} The $i$-th worker computes $\sum_{j=1}^k g_{ij}^t \Wbf_{j,t}$.

\STATE {\bfseries Output:} The master node outputs $\Xbf_T$.
\end{algorithmic}
\end{algorithm}

\subsection{Computing Gradient Descent using Substitute Decoding}\label{sec:gdcoding}

The substitute decoding method can also be applied to the distributed computing of gradient descent with erasures. Notice that gradient-coding has been studied substantially by previous works \cite{tandon2016gradient,raviv2017gradient,charles2017approximate,ye2018communication,halbawi2017improving}. Some of these works also consider applying sparse encoding matrices to obtain better scalability to more erasures. Comparing to previous results which focus on coded computing of the SUM function, the substitute decoding method maintains a full copy of all the partial gradients from each worker for the purpose of substitution. From this angle, substitute decoding can be viewed as introducing momentum into the gradient computation using coding-based methods.

Suppose the whole dataset is split into $k$ non-overlapping subsets and distributed to $P$ workers according to the sparsity pattern matrix $\Gbf$, i.e., the $j$-th copy is sent to the $i$-th worker if $\Gbf_{i,j}=1$. Each worker obtains a constant number of subsets. At each iteration, each worker computes a constant number of partial gradients $\w_{j,t}$ for all the subset indices $j$ that are allocated to it. Then, it encodes these gradients by
\begin{equation}\label{eqn:zbf_i_GC}
\z_{i,t}=\sum_{j=1}^k g_{ij}^t \w_{j,t},
\end{equation}
and sends the encoded partial gradient to the master node. If all the encoded partial gradients are available, the master node obtains
\begin{equation}
\left[\begin{matrix}
  \z_{1,t}\\
  \vdots\\
  \z_{P,t}
\end{matrix}\right]=\left[\begin{matrix}
  g_{11}^t & \ldots & g_{1k}^t\\
  g_{21}^t & \ldots & g_{2k}^t\\
  \vdots  & \ddots & \vdots\\
  g_{P1}^t & \ldots & g_{Pk}^t
\end{matrix}\right]
\left[\begin{matrix}
  \w_{1,t}\\
  \vdots\\
  \w_{k,t}
\end{matrix}\right]=\Gbf^{(t)} \bar{\Wbf}_t,
\end{equation}
where $\bar{\Wbf}_t:=\left[\begin{matrix}
  \w_{1,t}\\
  \vdots\\
  \w_{k,t}
\end{matrix}\right]$. Then, the master node conducts the substitute decoding and computes the estimate of the full gradients
\begin{equation}\label{eqn:substitute_decoding_gradient_descent}
\widehat{\Wbf}_{t}=\Vbf_t\Vbf_t^\top\bar{\Wbf}_t+\widetilde{\Vbf}_t\widetilde{\Vbf}_t^\top\widehat{\Wbf}_{t-1},
\end{equation}
where $\widehat{\Wbf}_t:=\left[\begin{matrix}
  \widehat{\w}_{1,t}\\
  \vdots\\
  \widehat{\w}_{k,t}
\end{matrix}\right]$ are the estimates of the full gradients at the master node, and the matrices $\Vbf_t$ and $\wt\Vbf_t$ are obtained using SVD in a similar way as in the previous sections. The coded gradient descent algorithm is presented in Algorithm~\ref{alg:gc}.

\begin{algorithm}[!h]
   \caption{Computing Gradient Descent Using Substitute Decoding}\label{alg:gc}
\begin{algorithmic}
	\STATE {\bfseries Input:} The whole dataset and sparsity pattern matrix $\Gbf$.
    \STATE {\bfseries Preprocessing:} Partition the dataset into subsets and store them distributedly as specified by the sparsity pattern $\Gbf$. Generate a series of random generator matrices $\Gbf^{(t)},t=1,2,\ldots,T$.

   \STATE {\bfseries Master Node:} Send out $\x_t$ at each iteration and receive partial coded gradients $\z_{i,t}$ in \eqref{eqn:zbf_i_GC}. Decode the gradients using substitute decoding defined in \eqref{eqn:substitute_decoding_gradient_descent}.

   After that, the master node computes the sum of gradients $\sum_{j=1}^k \widehat{\w}_{j,t}$ and updates the parameter
   \begin{equation}
        \x_t = \x_{t-1}-\epsilon \sum_{j=1}^k \widehat{\w}_{j,t}.
   \end{equation}
	\STATE {\bfseries Workers:} Worker-$i$ computes $\z_{i,t}$ in \eqref{eqn:zbf_i_GC} and sends $\z_{i,t}$ to the master node.

\STATE {\bfseries Output:} The master node outputs $\x_T$ after $T$ iterations.
\end{algorithmic}
\end{algorithm}

\section{Simulation Results on Coded Iterative Computation using Substitute Decoding}\label{sec:simulation_all}

\subsection{Coded Power Iterations for PageRank Computation on the Twitter Graph}\label{sec:simulation}

\subsubsection{The row-splitting case}

To support Theorem~\ref{thm:converge}, we compare uncoded, replication-based power iterations and Algorithm~\ref{alg:pg} on the Twitter graph \cite{leskovec2012learning}. We also show the result of noiseless power iterations. We run 100 independent simulations and average the results. There are $P=20$ workers. In each iteration, 50\% of the workers are disabled randomly. In the uncoded simulation, the graph matrix is partitioned into $P=20$ row blocks. The master node updates $\x_{t+1}=\Bbf\x_t+\y$ on the row blocks where results are available, and maintains the unavailable rows as $\x_t$. In the replication-based simulation, $\Bbf$ is partitioned into 10 row blocks and each one is replicated in 2 workers. Therefore, in each iteration, effectively 50\% of the entries in $\x_t$ get updated in the uncoded simulation and about 75\% of the entries in $\x_t$ get updated in the replication-based simulation. For the coded case, the sparsity pattern matrix $\Gbf$ is randomly generated using Definition~\ref{def:cyclic} with degree $d=2$ and $d=3$. The code is a $(20,10)$ code with rate $1/2$. We show in Table \ref{table:ds} how the sample average estimate of $\delta_t$ changes with the degree of $\Gbf$.
\begin{figure}
	\centering
	\includegraphics[scale=0.4]{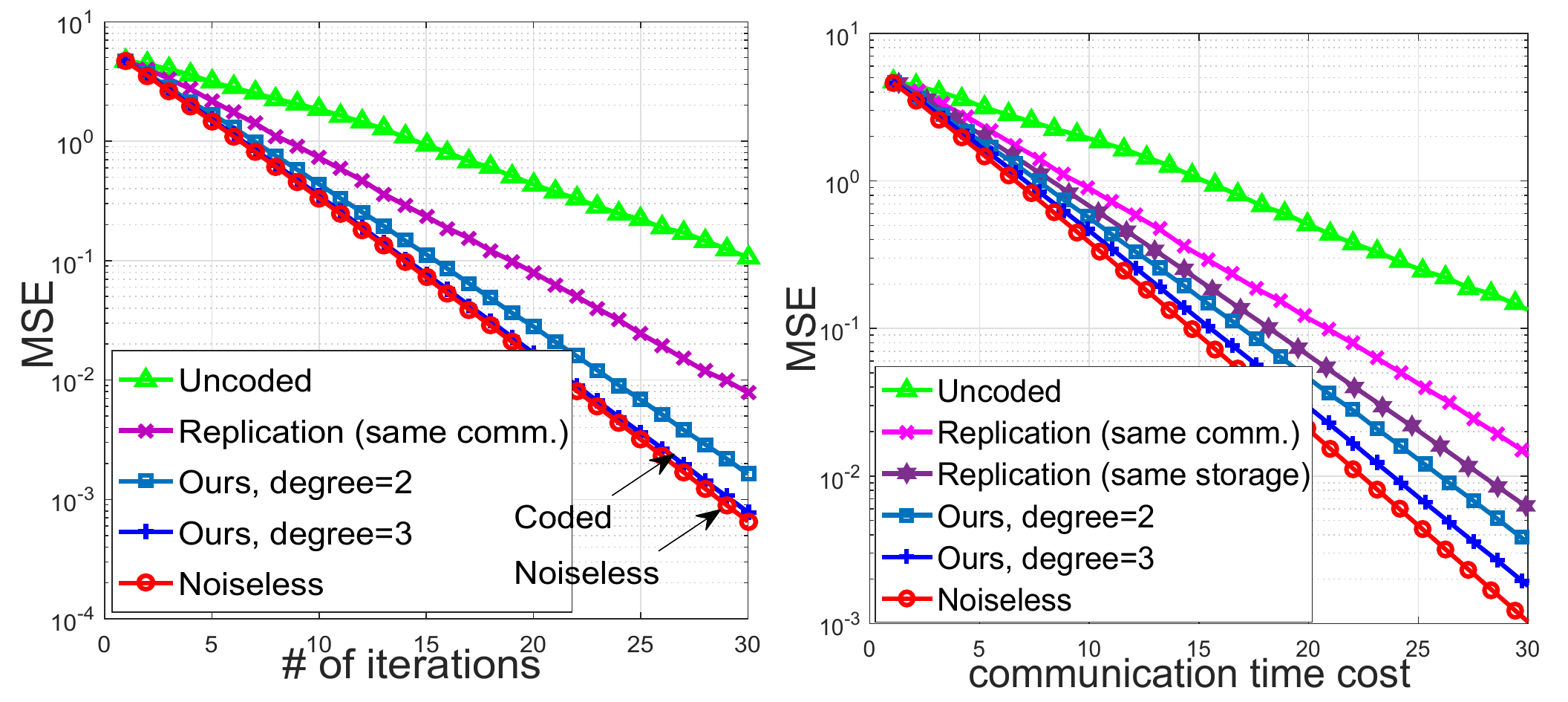}\\
	\caption{The comparison between uncoded, replication-based and substitute-decoding-based power iterations on the Twitter graph. Substitute decoding (\textcolor{blue}{blue line}) achieves almost exactly the same convergence rate as the noiseless case (\textcolor{red}{red line}) for the same number of iterations. Coded computing also beats the other techniques for the same communication time complexity.}\vspace{-2mm} \label{fig:simulation}
\end{figure}

{\bf Cost Analysis:} We also compare the convergence rates against communication cost (see Fig.~\ref{fig:simulation}; right). For Algorithm~\ref{alg:pg}, $\Bbf$ is partitioned into $k=10$ row blocks and encoded into 20 row blocks. The communication complexity in each iteration is $N(1+1/k)=1.1 N$. Similarly, it can be shown that the communication complexity of uncoded and replication-based power iterations are respectively $N(1+1/P)=1.05 N$ and $N(1+1/k)=1.1 N$. Since the average degree of the sparsity pattern is $d=2\sim 3$, computation cost and memory consumption only increase by a constant. We also plot the tradeoff for replication scheme with the same storage cost as $d=3$. However, in this case, the communication complexity for replication is larger, which is $N(1+d/k)=1.3 N$.
\begin{table}
\centering
\begin{tabular}{c|cccc}
\hline
$\bar{d}(\Gbf)$ & 2 & 3 & 4 & 5\\
\hline
$\delta_t$ & 0.1294 & 0.0442 &0.0243 & 0.0040\\
\hline
\end{tabular}
\caption{The factor $\delta_t$ decreases when the degree of $\Gbf$ increases.\vspace{-5mm}}
\label{table:ds}
\end{table}

\subsubsection{The column-splitting case}

\begin{figure}
  \centering
  \includegraphics[scale=0.45]{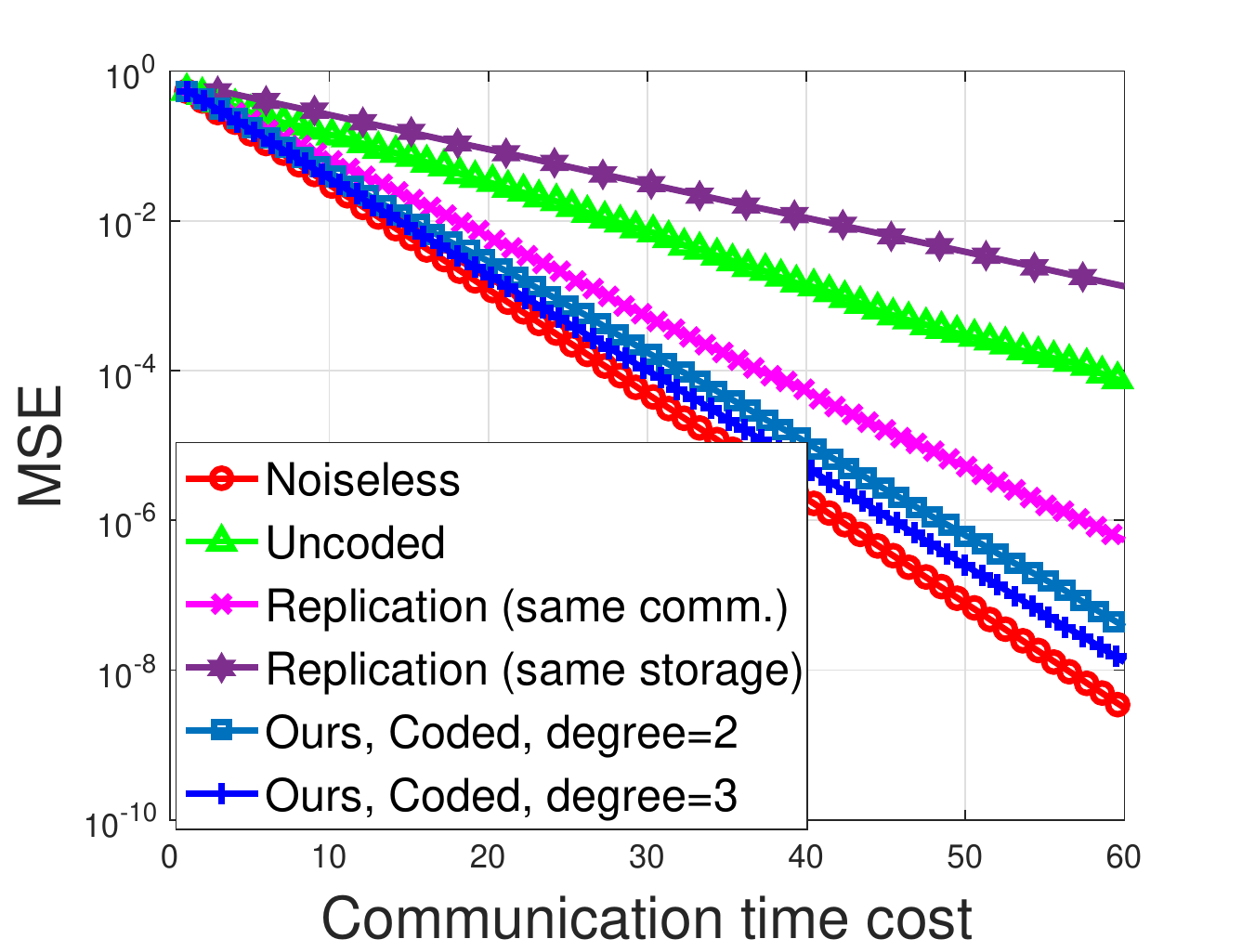}\\
  \caption{The comparison between uncoded, replication-based and substitute-decoding-based power iterations in column-wise splitting. Substitute decoding (\textcolor{blue}{blue line}) achieves almost exactly the same convergence rate as the noiseless case (\textcolor{red}{red line}). Coded computing beats the other techniques for the same communication time complexity.} \label{fig:simulation_col}
\end{figure}

To support Theorem~\ref{thm:converge_col}, we compare Algorithm~\ref{alg:pg_col} with the baseline methods again on the Twitter graph, but with column splitting (see Fig.~\ref{fig:simulation_col}). We split the graph into $k=48$ column-blocks and encoded them into $P=96$ column-blocks. Notice that compared to the row-splitting case, the coded power iterations in the column-splitting case have higher communication cost because each worker has to receive either 2 or 3 subvectors from the master node depending on the degree of the generator matrix (number of non-zeros in each row). This makes its communication complexity increase from $N(1+1/k)$ to $N(1+d/k)$, and $d=2$ or $3$.

\subsubsection{The SUMMA-splitting case}\label{sec:SUMMA_simulation}
\begin{figure}
  \centering
  \includegraphics[scale=0.45]{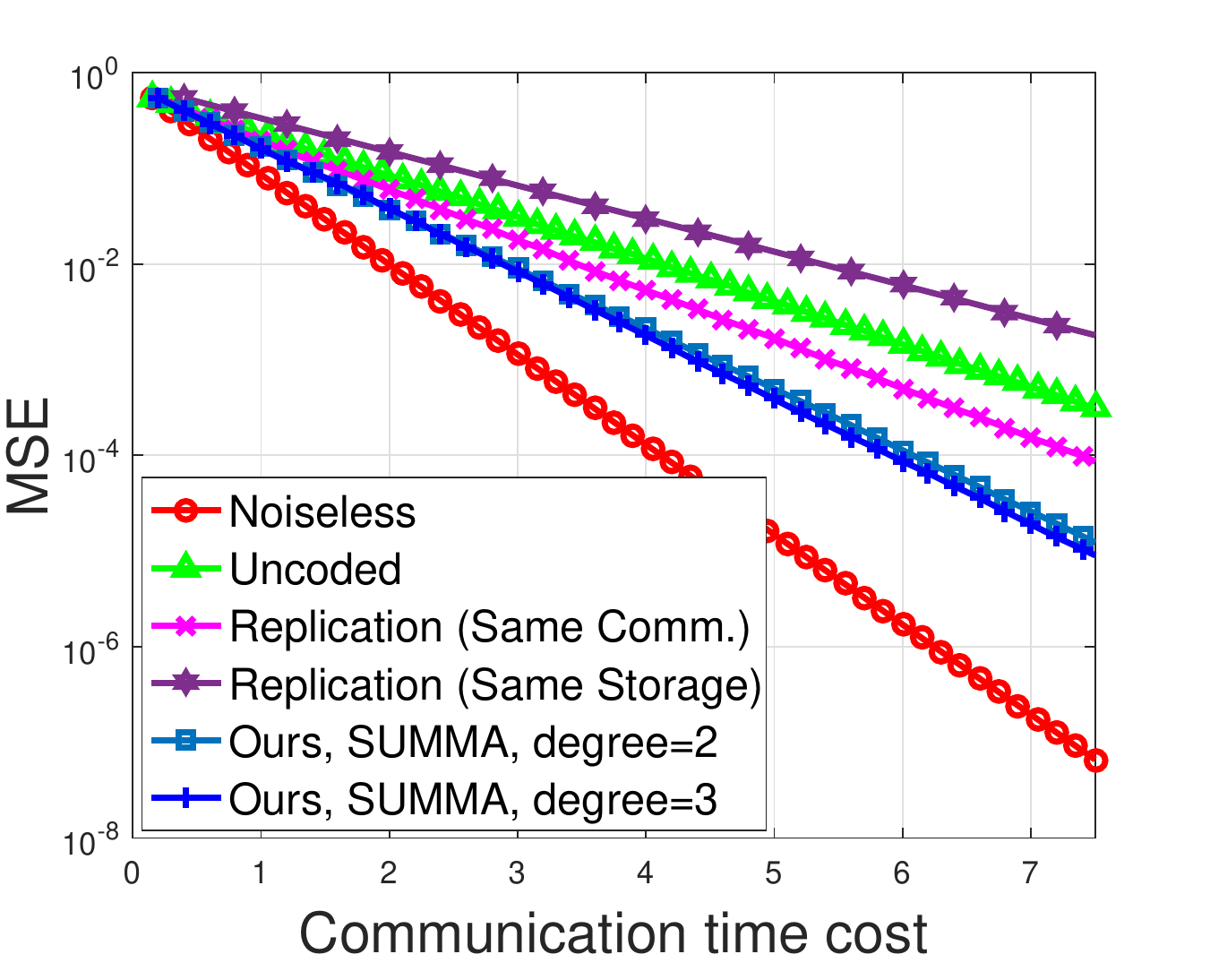}\\
  \caption{The comparison between uncoded, replication-based and substitute-decoding-based power iterations in SUMMA splitting. All schemes use SUMMA splitting on the linear system matrix. Substitute decoding with degree 3 and degree 2 (two \textcolor{blue}{blue} lines) are not close to the noiseless case (\textcolor{red}{red line}) because of increased communication time cost due to low-rate coding (with rate 1/2). Coded computing still beats the other techniques for the same communication time cost.} \label{fig:simulation_SUMMA}
\end{figure}

In SUMMA splitting, we also compare uncoded, replication-based power iterations and the coded power iterations on the Twitter graph (see Fig.~\ref{fig:simulation_SUMMA}). The number of workers is $P=200$ and the matrix is partitioned into $k=100$ square submatrices using a $\sqrt{k}\times \sqrt{k}=10\times 10$ SUMMA splitting. The workers are grouped into subsets of $\sqrt{k}=10$ and in each subset, we apply substitute decoding with a $(P/\sqrt{k},\sqrt{k})=$($20$,$10$) code. The other experimental setting is exactly the same as in the row-splitting case.
\begin{remark}
Notice that the number of workers that respond in one of the worker subsets may be different from that in another subset, and this is the reason that the coded computing cannot be expected to achieve the convergence rate of the noiseless case. This is also why the performance of coded computing seems to saturate when the degree changes from 2 to 3 in Figure~\ref{fig:simulation_SUMMA}. More specifically, the substitute decoding can \emph{provide the full advantage} when the number of responses is equal to $\sqrt{k}$ in each group. When the number of responses in a group is larger than $\sqrt{k}$, the partial encoding matrix inside that particular group has larger than $\sqrt{k}$ rows, but it can only provide similar results with the case when only $\sqrt{k}$ workers respond because the partial encoding matrix in the latter case is already close to full rank (see Table \ref{table:ds} for the small $\delta_t$ when the number of responses is equal to the number of splits). Thus, to achieve improved performance for a sparse code with a larger degree, it may be useful to design a coding algorithm that uses entangled encoding in both rows and columns of the SUMMA splitting. However, the most straightforward way of using entangled row and column coding, i.e., directly combining the coded computing schemes in Algorithm~\ref{alg:pg} and Algorithm~\ref{alg:pg_col}, will increase the communication cost substantially and is not suitable for the purpose of reducing speed in distributed computing.
\end{remark}

\subsection{Coded Orthogonal Iterations for Spectral Clustering}\label{sec:spec_clus_simu}
\begin{figure}
  \centering
  \includegraphics[scale=0.45]{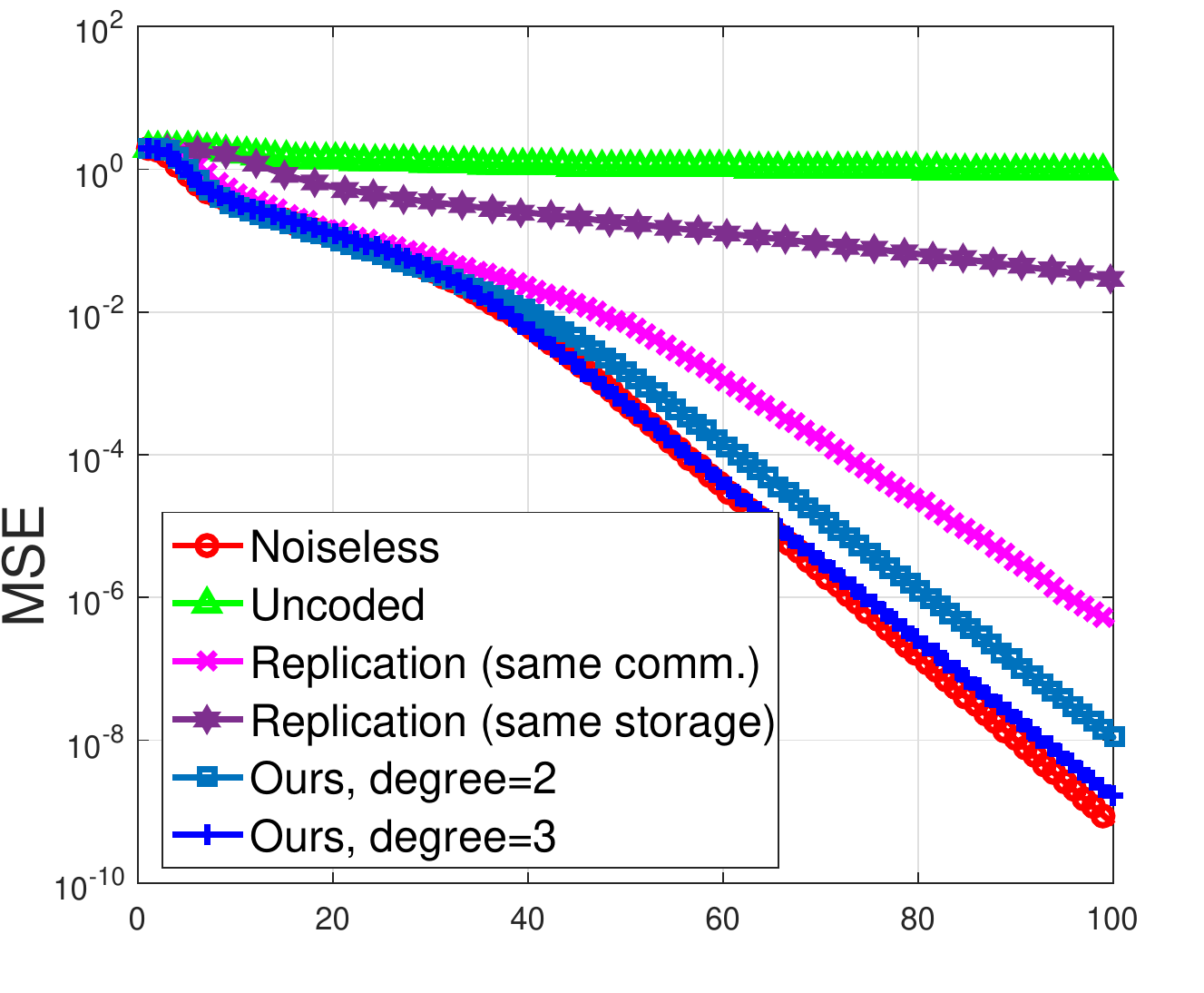}\\
  \caption{The comparison between uncoded, replication-based and substitute-decoding-based orthogonal iterations. Substitute decoding (\textcolor{blue}{blue line}) achieves almost exactly the same convergence rate as the noiseless case (\textcolor{red}{red line}).}\label{fig:simulation2}
\end{figure}
\begin{figure}
  \centering
  \includegraphics[scale=0.45]{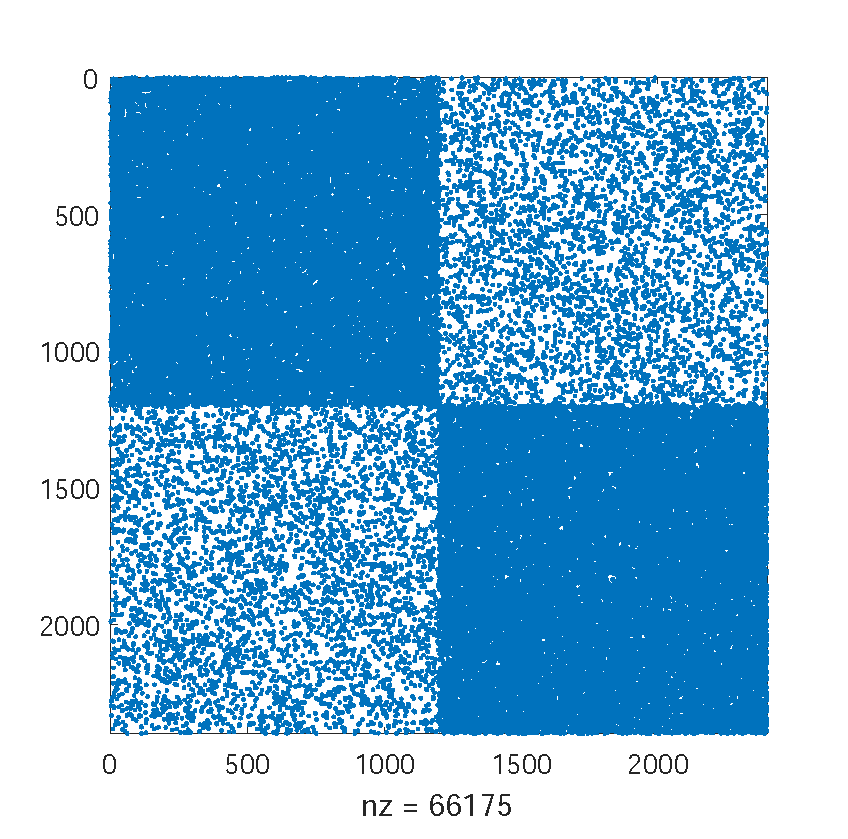}\\
  \caption{This figure shows the clustering result of the graph adjacency matrix using the spectral clustering algorithm with coded computing techniques.}\label{fig:clustering}
\end{figure}

To test the performance of Algorithm~\ref{alg:oi}, we compare uncoded, replication-based orthogonal iterations and Algorithm~\ref{alg:oi} in the application of spectral clustering \cite{von2007tutorial,ng2002spectral} on synthesized graphs (see Fig.~\ref{fig:simulation2}) generated from the stochastic block model with two clusters. We also show the result of noiseless orthogonal iterations. We compute the first two eigenvectors of the normalized graph Laplacian matrix and measure convergence in terms of the MSE of the eigenvector estimation. The second eigenvector is used to generate the clustering result in Fig.~\ref{fig:clustering} with a threshold value 0, i.e., the nodes are partitioned into two clusters based on the signs of the corresponding entries in the second eigenvector.

We generate a graph from the stochastic block model with two clusters and with intra-cluster connection probability 0.02 and inter-cluster connection probability 0.003. There are $P=96$ workers. In each iteration, 50\% of the workers are disabled randomly. In the uncoded simulation, the graph matrix is partitioned into $P=96$ column blocks. The master node updates $\Zbf_{t}=\Bbf\Xbf_t$ using the column blocks where results are available, and maintains the unavailable column blocks from the last iteration. In the replication-based simulation (same communication), $\Bbf$ is partitioned into $k=48$ column blocks and each one is replicated in 2 workers. The replication-based method (same storage) uses the same storage of data as the coded case but does not compute the linear combinations of the partial results. For the case of coded computing, the sparsity pattern matrix $\Gbf$ is randomly generated using Definition~\ref{def:cyclic} for degree $d=2$ and $d=3$. The code is a $(96,48)$ code with rate $1/2$. We run 100 independent simulations and average the results. Since the number of eigenvectors to compute is very small, we do not use the acceleration method in \eqref{eqn:acc21}-\eqref{eqn:acc23}. In this case, the uncoded computation does not converge at all.

{\bf Cost Analysis:} In each iteration, the communication complexity is $Nr(1+d/k)=1.04 Nr$ or $1.06Nr$ because we have $r$ eigenvectors to compute. Similarly, it can be shown that the communication complexity of uncoded, replication-based (same communication) and replication-based (same storage) orthogonal iterations are respectively $Nr(1+1/P)=1.01 Nr$, $Nr(d+d/k)=3.06 Nr$ (we use $d$=3 for replication with the same storage) and $Nr(1+1/k)=1.02 Nr$. Therefore, the communication costs of these strategies are similar.

\subsection{Coded Orthogonal Iterations for Singular Value Decomposition}\label{sec:pca_simu}
\begin{figure}
  \centering
  \includegraphics[scale=0.45]{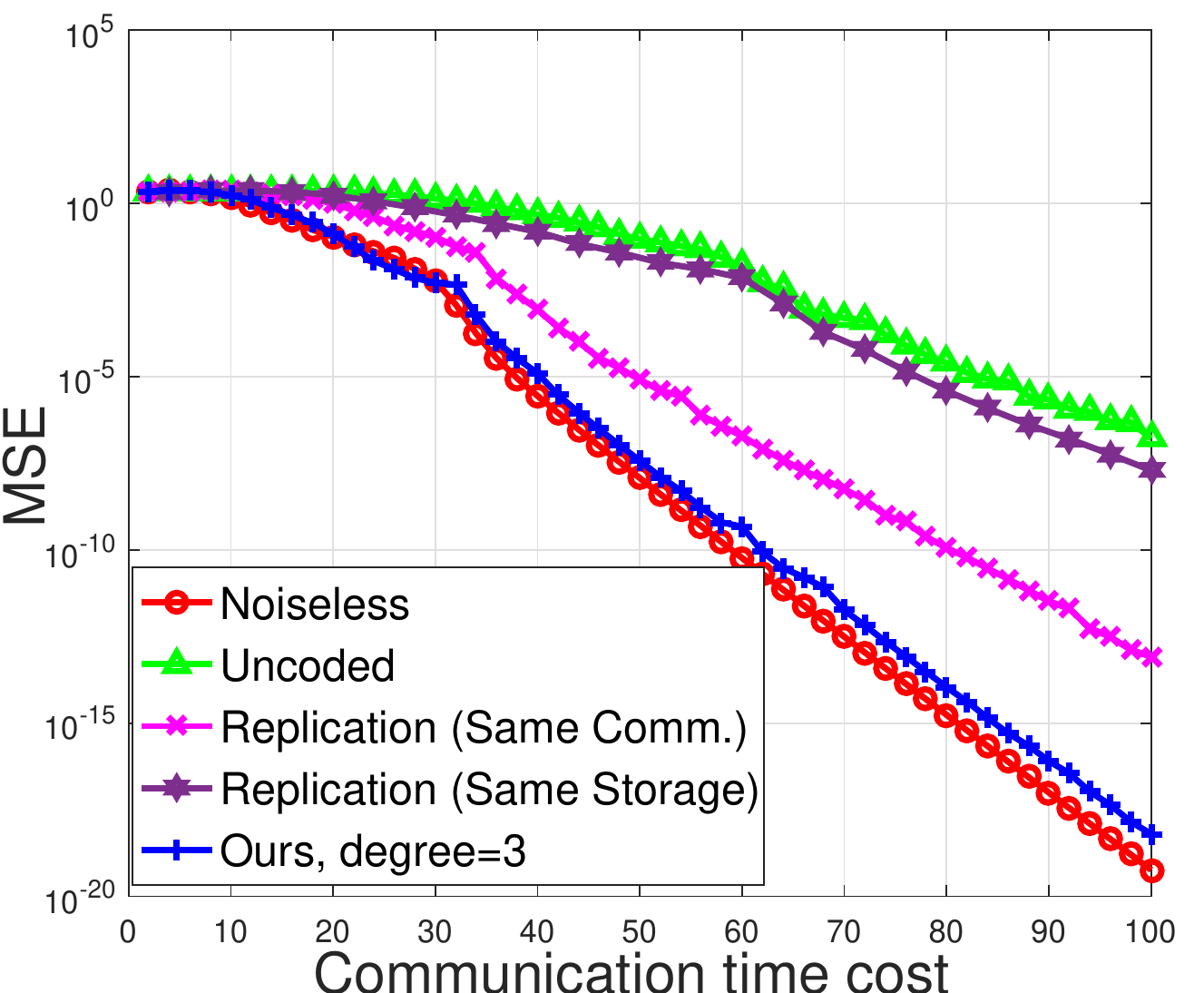}\\
  \caption{The comparison between uncoded, replication-based and substitute-decoding-based orthogonal iterations for principal component analysis. Substitute decoding (\textcolor{blue}{blue line}) achieves almost exactly the same convergence rate as the noiseless case (\textcolor{red}{red line}). Coded computing beats the other techniques for the same communication time complexity.}\label{fig:simulation3}
\end{figure}

\begin{figure}
  \centering
  \includegraphics[scale=0.45]{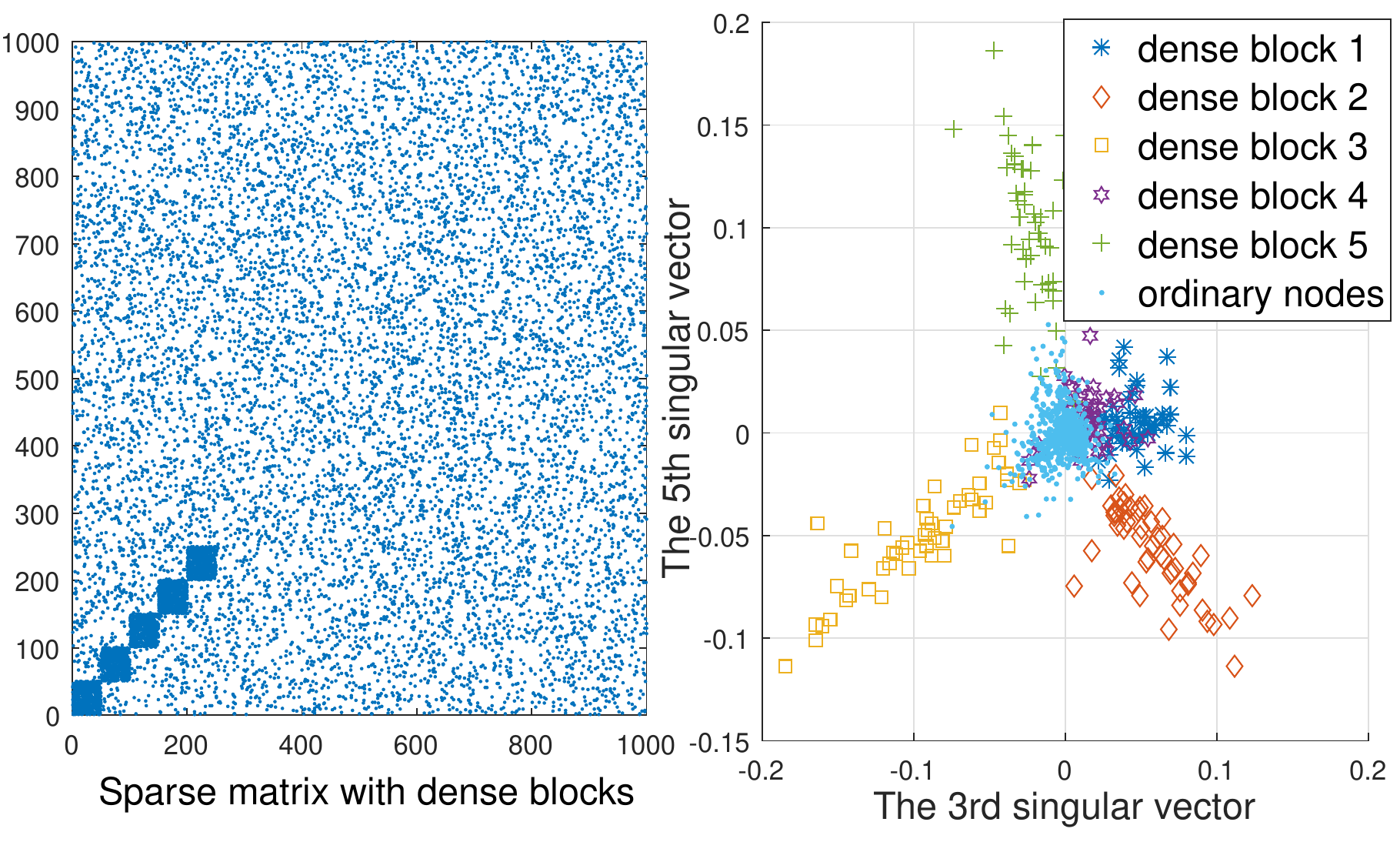}\\
  \caption{This figure shows the phenomenon of ``Eigenspokes'', which shows that the principal components of a sparse matrix with dense blocks can have spoke-like patterns\cite{prakash2010eigenspokes}. These patterns can help identify anomalous dense blocks inside a huge network.}\label{fig:eigen_spoke}
\end{figure}

We compare Algorithm~\ref{alg:pca} with the baseline algorithms on synthesized matrices with planted dense submatrices (see Fig.~\ref{fig:eigen_spoke}; left) and compute the first 5 singular vectors. We report the MSE of the computed eigenvectors in Fig.~\ref{fig:simulation3}. We also show the result of noiseless orthogonal iterations. We run 100 independent simulations and average the results. In each simulation, we generate a random sparse matrix of size $1000\times 1000$ with non-zero probability 0.01 and plant 5 dense blocks of size $50\times 50$ with non-zero probability 0.2. Each non-zero entry is uniformly distributed in $[0,1]$. There are $P=100$ workers. In each iteration, 50\% of the workers are disabled randomly. The noiseless case is the same as the one described in Section~\ref{sec:pca_noiseless}. In the uncoded simulation, the master node maintains the partial computation result $\Zbf_{i,t}=\Bbf_i^\top \Bbf_i\Xbf_t$ from the $i$-th worker at each time slot. If at the $(t+1)$-th time slot the $i$-th worker fails to send the result, the master node just uses the same partial result in the last iteration. In the replication-based simulation (same communication), $\Bbf$ is partitioned into $k=50$ row blocks and each one is replicated in 2 workers. The replication-based method (same communication) is similar to the uncoded one except that each $\Zbf_{i,t}$ is computed in two workers for the purpose of achieving fault/straggler tolerance. The replication-based method (same storage) uses the same storage of data as the coded case but does not compute the linear combinations of the partial results. For the coded case, the sparsity pattern matrix $\Gbf$ is randomly generated by assigning 3 ones in each row. The code is a $(100,50)$ code with rate $1/2$. We also report the phenomenon of ``Eigenspokes'' in the right part of Fig. \ref{fig:eigen_spoke}. For each scattered point ($x,y$), $x$ is the corresponding entry in the 3rd singular vector, and $y$ is the corresponding entry in the 5th singular vector. The scattered points in the anomalous dense blocks show regular patterns on this plot.

{\bf Cost Analysis:} The communication complexity is $2Nr$ for all schemes except the replication scheme with the same storage, in which case the communication complexity is $2(1+d)Nr$, where $d$ is the number of ones assigned to each row of the encoding matrix. The computation cost of coded computing increases by a constant factor compared to the uncoded case.

\subsection{Coded Gradient Descent using Substitute Decoding}\label{sec:gc_simu}

\begin{figure}
  \centering
  \includegraphics[scale=0.45]{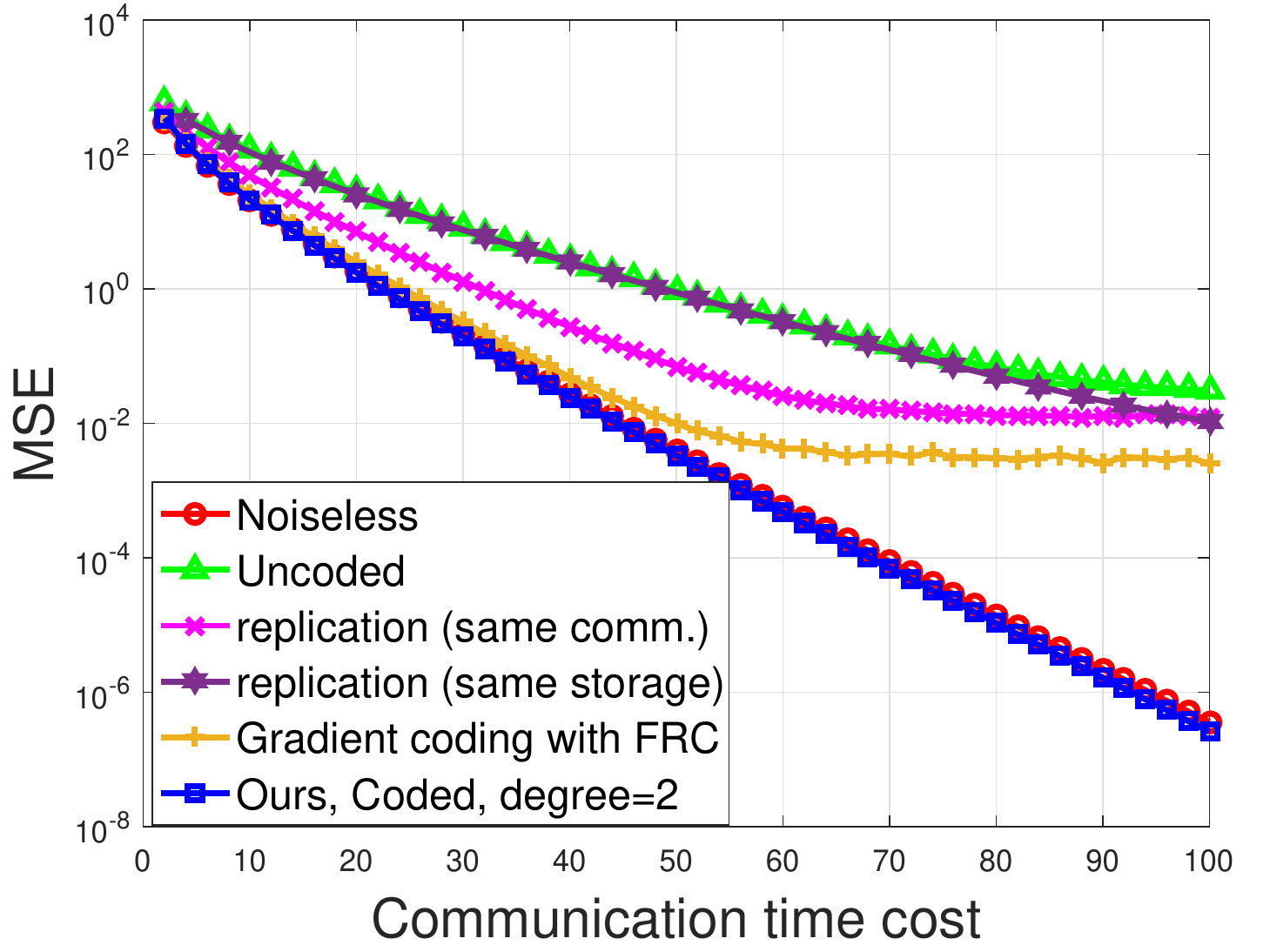}\\
  \caption{The comparison between uncoded, replication-based, approximate-gradient-coding-based \cite{charles2017approximate} and substitute-decoding-based gradient-descent computing. Substitute decoding (\textcolor{blue}{blue line}) achieves exactly the same convergence rate as noiseless computation (\textcolor{red}{red line}). Coded computing beats the other techniques for the same communication time complexity. The reason that the coded computing converges slightly faster than the noiseless case is explained in Remark \ref{rem:faster_than_noiseless}.}\label{fig:simulation4}
\end{figure}

To test the performance of Algorithm~\ref{alg:gc}, we compare uncoded, replication-based, and approximate-gradient-coding-based \cite{charles2017approximate} gradient computing (using fractional repetition codes) and Algorithm~\ref{alg:gc} on synthesized data (see Fig.~\ref{fig:simulation4}). We also show the result of noiseless gradient descent.

We compute the result of the following optimization problem
\begin{equation}
\min_\x \twonorm{\y-\Abf\x}_2,
\end{equation}
where $\Abf$ is the data matrix of size 5000$\times$1000 and $\y$ is of length 5000. We compute $\x$ using the vanilla gradient descent
\begin{equation}
\x_{t+1} = \x_t - \epsilon \frac{1}{N}\Abf^\top(\Abf \x_t-\y).
\end{equation}

We run 100 independent simulations and average the results. In each simulation, we generate a random Gaussian matrix $\Abf$ and a Gaussian random vector $\y$. There are $P=100$ workers. All schemes use the same step size $\epsilon=0.5$. In the noiseless case, we partition the dataset into 100 parts and store one part at each worker. All workers can successfully compute the correct results and send to the master. In the other cases, in each iteration, 50\% of the workers are disabled randomly. In the uncoded simulation, the master aggregates the gradients from only the workers that successfully send back the partial gradients. In the replication-based simulation (same communication), the data is partitioned into 50 subsets and each one is replicated in 2 workers. The replication-based method (same communication) is similar to the uncoded one except that each partial gradient is computed in two workers for the purpose of achieving fault/straggler tolerance. The replication-based method (same storage) uses the same storage of data as the coded case but does not compute the linear combinations of the partial results. For the coded case, the sparsity pattern matrix $\Gbf$ is randomly generated by assigning 2 ones in each row. The code is a $(100,50)$ code with rate $1/2$. The gradient coding algorithm that we compare with is the one in \cite{charles2017approximate} using fractional repetition code. The coding matrix is also of size $(100,50)$ and there are two ones in each row, i.e., each worker computes two partial gradients and transmits the sum, and each pair of partial gradients is computed at four workers. We choose this algorithm to compare with because it also computes the approximate gradient.

{\bf Cost Analysis:} The communication complexity in this case is $2d_\text{data}$ where $d_\text{data}$ is the dimension of the data, because we have two rounds of communication during each iteration. The communication costs of all the compared schemes are the same except the replication scheme with the same storage, in which case the communication complexity is $2(1+d)d_\text{data}$. The computation cost of the coded method increases by a constant factor compared to the uncoded method and replication-based method (same communication). The gradient coding method, the replication method (same storage) and our algorithm have exactly the same communication cost, computation cost and storage cost.
\begin{remark}\label{rem:faster_than_noiseless}
It may be surprising that the result of coded computing actually converges slightly faster than the noiseless case. Our explanation is that the coded computing with substitute decoding provides a way of combining past gradients with the current gradients and hence introduces a certain type of momentum into the computation of gradient descent. It has been observed for long \cite{rumelhart1986learning} that introducing momentum may prevent the convergence trajectory from oscillating in a ``narrow valley''. Therefore, we conjecture that for the specific problem of computing gradient descent, the substitute decoding method provides a coded way of introducing momentum. Deeply investigating this behavior and the possible improvements in coding schemes resulted from this behavior is our future goal.
\end{remark}

\section{Proofs}\label{sec:proof_all}
\subsection{Proof of Theorem~\ref{thm:converge}}\label{sec:proof}
\begin{lemma}\label{lmm:exp_proj}
If the sparsity pattern in Definition~\ref{def:cyclic} is used and Assumption~\ref{ass:noise} holds, the projection matrix $\Vbf_t\Vbf_t^\top$ satisfies
\begin{equation}\label{eqn:ds}
\Ep[\Vbf_t\Vbf_t^\top]=(1-\delta_t)\I_k,
\end{equation}
where the expectation is taken respect to the randomness of non-zero entries' values (the sparsity pattern $\Gbf$ is fixed) and the randomness of the workers' failure events.
\end{lemma}
\begin{proof}
See Appendix~\ref{sec:lmm_proof} for the full proof. The proof relies on proving some symmetric properties of $\Ep[\Vbf_t\Vbf_t^\top]$. We prove that the symmetry of standard Gaussian pdf on the real line ensures that all of the off-diagonal entries in $\Ep[\Vbf_t\Vbf_t^\top]$ are zero. Further, we prove that the ``combined cyclic'' structure of $\Gbf$ in Definition~\ref{def:cyclic} ensures that all diagonal entries on $\Ep[\Vbf_t\Vbf_t^\top]$ are identical. The two facts above show that $\Ep[\Vbf_t\Vbf_t^\top]=x\I_k$ for some constant $x$. Then, we can use a property of trace to compute $x$ and obtain \eqref{eqn:ds}.
\end{proof}
We denote the projection $\Vbf_t\Vbf_t^\top$ by $\Pbf_V$. Then, from \eqref{eqn:sum_vtilde_I}, $\wt\Vbf_t\wt\Vbf_t^\top=\I_k-\Pbf_V$. The first term $\vect{\Vbf_t\Vbf_t^\top\mat{\Bbf \e_t}}=\vect{\Pbf_V\mat{\Bbf \e_t}}$ in \eqref{eqn:err_iter} can be bounded by
\[
\begin{split}
&\Ep\left[\twonorm{\vect{\Pbf_V\mat{\Bbf \e_t}}}^2\right]
\overset{(a)}{=}\Ep\left[\twonorm{(\Pbf_V\otimes \I_b)\Bbf \e_t}^2\right]\\
\overset{(b)}{=}&\Ep\left[\trace\left((\Bbf \e_t)^\top(\Pbf_V\otimes \I_b)^\top (\Pbf_V\otimes \I_b)\Bbf \e_t \right)\right]\\
\end{split}
\]
\begin{equation}\label{eqn:der1}
\begin{split}
\overset{(c)}{=}&\Ep\left[\trace\left(\Bbf \e_t(\Bbf \e_t)^\top(\Pbf_V\otimes \I_b)^\top (\Pbf_V\otimes \I_b) \right)\right]\\
\overset{(d)}{=}&\trace\left(\Ep\left[\Bbf \e_t(\Bbf \e_t)^\top\right]\Ep\left[(\Pbf_V\otimes \I_b)^\top (\Pbf_V\otimes \I_b) \right]\right)\\
=&\trace\left(\Ep\left[\Bbf \e_t(\Bbf \e_t)^\top\right]\Ep\left[(\Pbf_V^\top \Pbf_V)\otimes \I_b \right]\right)\\
\overset{(e)}{=}&\trace\left(\Ep\left[\Bbf \e_t(\Bbf \e_t)^\top\right]\Ep\left[\Pbf_V\otimes \I_b \right]\right)\\
\overset{(f)}{=}&\trace\left(\Ep\left[\Bbf \e_t(\Bbf \e_t)^\top\right]((1-\delta_t)\I_k)\otimes \I_b \right)\\
=& (1-\delta_t)\trace\left(\Ep\left[\Bbf \e_t(\Bbf \e_t)^\top\right]\right)=(1-\delta_t)\Ep[\twonorm{\Bbf \e_t}^2],
\end{split}
\end{equation}
where ($a$) is from the property of mat-vec operations, ($b$) is because $(\Pbf_V\otimes \I_b)\Bbf \e_t$ is a vector, ($c$) is because $\trace(AB)=\trace(BA)$, ($d$) is because $\trace$ and $\Ep$ commutes and the projection $\Pbf_V$ only depends on the random partial generator matrix $\Gbf_s$ and is independent of $\e_t$, ($e$) is because $\Pbf_V$ is a projection matrix and ($f$) is from Lemma~\ref{lmm:exp_proj} and $\Pbf_V=\Vbf_t\Vbf_t^\top$. Similarly, we can prove
\begin{equation}\label{eqn:der2}
\Ep\left[\twonorm{\vect{(\I_k-\Pbf_V)\mat{\e_t}}}^2\right]=\delta_t\Ep[\twonorm{\e_t}^2].
\end{equation}
Therefore
\begin{equation}
\begin{split}
&\Ep[\twonorm{\e_{t+1}}^2]\\
\overset{(a)}{=}&\Ep[\twonorm{\vect{\Vbf_t\Vbf_t^\top\mat{\Bbf \e_t}}}^2]+\Ep[\twonorm{\vect{\wt\Vbf_t\wt\Vbf_t^\top\mat{\e_t}}}^2]\\
\overset{(b)}{=}&(1-\delta_t) \Ep[\twonorm{\Bbf\e_{t}}^2]+\delta_t\Ep[\twonorm{\e_{t}}^2],
\end{split}
\end{equation}
where ($a$) is from \eqref{eqn:err_iter} and the Pythagorean theorem, and ($b$) is from \eqref{eqn:der1} and \eqref{eqn:der2}. Thus, we have completed the proof.

\subsection{Proof of Theorem~\ref{thm:converge_col}}

Since $\Vbf_t$ and $\wt\Vbf_t$ are orthogonal to each other we have
\begin{equation}\label{eqn:error_analysis_2}
\begin{split}
  [\u^{1*},\ldots,\u^{k*}]=&[\u^{1*},\ldots,\u^{k*}]\Vbf_t\Vbf_t^\top+[\u^{1*},\ldots,\u^{k*}]\wt\Vbf_t\wt\Vbf_t^\top\\
  =&[\Bbf_1\x^{1*},\ldots,\Bbf_k\x^{k*}]\Vbf_t\Vbf_t^\top\\
  &+[\u^{1*},\ldots,\u^{k*}]\wt\Vbf_t\wt\Vbf_t^\top.
\end{split}
\end{equation}
Subtracting \eqref{eqn:error_analysis_2} from \eqref{eqn:substitute_decoding_vertical}, we have
\begin{equation}\label{eqn:linsys_1}
\begin{split}
[\e_t^1,\ldots, \e_t^k]=&[\Bbf_1\ep_t^1,\ldots,\Bbf_k\ep_t^k]\Vbf_t\Vbf_t^\top\\
&+[\e_{t-1}^1,\ldots,\e_{t-1}^k]\wt\Vbf_t\wt\Vbf_t^\top,
\end{split}
\end{equation}
where $\ep_t^j=\x_t^j-\x^{j*}$ is the $j$-th subvector of $\e_t=\x_t-\x^*$. The above equation and the following equation (from \eqref{eqn:e_sum}) describe two coupled linear systems of $\ep_t^j$ and $\e_t^j$:
\begin{equation}\label{eqn:linsys_2}
  \left[\begin{matrix}
    \ep_{t+1}^1\\
    \vdots\\
    \ep_{t+1}^k
  \end{matrix}\right]= \e_{t+1} =\sum_{j=1}^k \e_t^j.
\end{equation}
Now, we analyze the coupled linear systems \eqref{eqn:linsys_1} and \eqref{eqn:linsys_2}. First, we vectorize the matrices in \eqref{eqn:linsys_1}. This gives the following equation
\begin{equation}\label{eqn:error_analysis_3}
\left[\begin{matrix}
  \e_t^1\\
  \vdots\\
  \e_t^k
\end{matrix}\right]=[\Vbf_t\Vbf_t^\top\otimes\I_N]\left[\begin{matrix}
  \Bbf_1\ep_t^1\\
  \vdots\\
  \Bbf_k\ep_t^k
\end{matrix}\right]+[\wt\Vbf_t\wt\Vbf_t^\top\otimes\I_N]\left[\begin{matrix}
  \e_{t-1}^1\\
  \vdots\\
  \e_{t-1}^1
\end{matrix}\right].
\end{equation}
Define
\begin{equation}
  \Bbf_\text{expand}=\left[\begin{matrix}
  \Bbf_1&&\\
  &\ddots&\\
  &&\Bbf_k
\end{matrix}\right].
\end{equation}
Then, from \eqref{eqn:error_analysis_3}, we have
\begin{equation}
\begin{split}
  \left[\begin{matrix}
  \e_t^1\\
  \vdots\\
  \e_t^k
\end{matrix}\right]=&[\Vbf_t\Vbf_t^\top\otimes\I_N]\Bbf_\text{expand}\left[\begin{matrix}
  \ep_t^1\\
  \vdots\\
  \ep_t^k
\end{matrix}\right]+[\wt\Vbf_t\wt\Vbf_t^\top\otimes\I_N]\left[\begin{matrix}
  \e_{t-1}^1\\
  \vdots\\
  \e_{t-1}^1
\end{matrix}\right]\\
\overset{(a)}{=}&[\Vbf_t\Vbf_t^\top\otimes\I_N]\Bbf_\text{expand}\sum_{j=1}^k \e_{t-1}^j
+[\wt\Vbf_t\wt\Vbf_t^\top\otimes\I_N]\left[\begin{matrix}
  \e_{t-1}^1\\
  \vdots\\
  \e_{t-1}^1
\end{matrix}\right]\\
=&[\Vbf_t\Vbf_t^\top\otimes\I_N]\Bbf_\text{expand}[\I_N,\I_N,\ldots,\I_N]\left[\begin{matrix}
  \e_{t-1}^1\\
  \vdots\\
  \e_{t-1}^k
\end{matrix}\right]\\
&+[\wt\Vbf_t\wt\Vbf_t^\top\otimes\I_N]\left[\begin{matrix}
  \e_{t-1}^1\\
  \vdots\\
  \e_{t-1}^1
\end{matrix}\right],
\end{split}
\end{equation}
where step ($a$) is from \eqref{eqn:linsys_2}.

Recall that
\begin{equation}
E_t=\left[\begin{matrix}
  \e_t^1\\
  \vdots\\
  \e_t^k
\end{matrix}\right].
\end{equation}
Then, the above equation can be simplified to
\begin{equation}
E_t=[\Vbf_t\Vbf_t^\top\otimes\I_N]\Bbf_\text{expand}^{(k)}E_{t-1}+[\wt\Vbf_t\wt\Vbf_t^\top\otimes\I_N]E_{t-1},
\end{equation}
where $\Bbf_\text{expand}^{(k)}:=\Bbf_\text{expand}[\I_N,\I_N,\ldots,\I_N]$. Since $\Vbf_t\Vbf_t^\top\otimes\I_N$ and $\wt\Vbf_t\wt\Vbf_t^\top\otimes\I_N$ are two projections that are orthogonal to each other, we have
\begin{equation}\label{eqn:error_analysis_6}
\begin{split}
\Ep[\twonorm{E_t}^2]=&\Ep\left[\twonorm{[\Vbf_t\Vbf_t^\top\otimes\I_N]\Bbf_\text{expand}^{(k)}E_{t-1}}^2\right]\\
&+\Ep\left[\twonorm{[\wt\Vbf_t\wt\Vbf_t^\top\otimes\I_N]E_{t-1}}^2\right].
\end{split}
\end{equation}
Then, we can use the same derivation in \eqref{eqn:der1} to show that
\begin{equation}\label{eqn:error_analysis_4}
\Ep\left[\twonorm{[\Vbf_t\Vbf_t^\top\otimes\I_N]\Bbf_\text{expand}^{(k)}E_{t-1}}^2\right]=(1-\delta_t) \twonorm{\Bbf_\text{expand}^{(k)}E_{t-1}}^2,
\end{equation}
and
\begin{equation}\label{eqn:error_analysis_5}
\Ep\left[\twonorm{[\wt\Vbf_t\wt\Vbf_t^\top\otimes\I_N]E_{t-1}}^2\right]=\delta_t \twonorm{E_{t-1}}^2.
\end{equation}
Plugging \eqref{eqn:error_analysis_4} and \eqref{eqn:error_analysis_5} into \eqref{eqn:error_analysis_6}, we obtain
\begin{equation}
\begin{split}
  \Ep[\twonorm{E_t}^2]=&(1-\delta_t) \twonorm{\Bbf_\text{expand}^{(k)}E_{t-1}}^2+\delta_t \twonorm{E_{t-1}}^2\\
\le &(1-\delta_t) \twonorm{\Bbf_\text{expand}^{(k)}}_2^2\twonorm{E_{t-1}}^2+\delta_t \twonorm{E_{t-1}}^2.
\end{split}
\end{equation}
Thus, to prove \eqref{eqn:converge_rate_col}, we only need to prove $\twonorm{\Bbf_\text{expand}^{(k)}}_2=\twonorm{\Bbf}_\text{col}$, where $\twonorm{\Bbf}_\text{col}$ is defined by $\twonorm{\Bbf}_\text{col}=\sqrt{k}\max_{j}\{\twonorm{\Bbf_j}_2\}$. To prove this, we notice that
\begin{equation}
\begin{split}
  \twonorm{\Bbf_\text{expand}^{(k)}}_2=&\twonorm{\Bbf_\text{expand}[\I_N,\I_N,\ldots,\I_N]}_2\\
  =&\sqrt{k}\twonorm{\Bbf_\text{expand}}_2\\
  =&\sqrt{k}\twonorm{\left[\begin{matrix}
  \Bbf_1&&\\
  &\ddots&\\
  &&\Bbf_k
\end{matrix}\right]}_2\\
=&\sqrt{k} \max\{\twonorm{\Bbf_j}_2\}=\twonorm{\Bbf}_\text{col}.
\end{split}
\end{equation}
Therefore, we have completed the proof.

\bibliographystyle{IEEEtran}
\bibliography{rough}

\appendices

\section{Proof of Lemma~\ref{lmm:exp_proj}}\label{sec:lmm_proof}
\begin{proof}
First, notice that although the SVD decomposition of $\Gbf_{s}^{(t)}$ in \eqref{eqn:G_svd} is not unique, the projection matrix $\Vbf_t\Vbf_t^\top$ is unique for a certain $\Gbf_{s}^{(t)}$, because it is the projection onto the row space of $\Gbf_{s}^{(t)}$.  Then, before we prove the lemma, we show another lemma that will be useful.
\begin{lemma}\label{lmm:PP}
Suppose $\Pbf$ is a $k\times k$ orthonormal matrix. Then, if the corresponding projection matrix of $\Gbf_{s}^{(t)}$ is $\Vbf_t\Vbf_t^\top$, the corresponding projection matrix of $\Gbf_{s}^{(t)}\Pbf$ is $\Pbf^\top\Vbf_t\Vbf_t^\top\Pbf$.
\end{lemma}
\begin{proof}
If $\Gbf_{s}^{(t)}$ has the SVD decomposition $\Gbf_{s}^{(t)}=\Ubf_t\Dbf_t\Vbf_t^\top$, then $\Gbf_{s}^{(t)}\Pbf=\Ubf_t\Dbf_t(\Vbf_t^\top\Pbf)$ is a valid SVD of $\Gbf_{s}^{(t)}\Pbf$, which means the projection matrix now should be $(\Vbf_t^\top\Pbf)^\top \Vbf_t^\top\Pbf=\Pbf^\top\Vbf_t\Vbf_t^\top\Pbf$.
\end{proof}

Now, we prove two properties of the expected projection matrix.

\subsection{Property 1: all off-diagonal entries in $\Ep[\Vbf_t\Vbf_t^\top]$ are zero} Suppose $\Pbf_i$ of size $k\times k$ is the diagonal matrix where all diagonal entries are 1 except the $i$-th diagonal entry is $-1$. When we multiply $\Gbf_{s}^{(t)}\Pbf_i$, we effectively flips the $i$-th column of $\Gbf_{s}^{(t)}$. Now, we observe the fact that the distribution of $\Gbf_{s}^{(t)}\Pbf_i$ is exactly the same as the distribution of $\Gbf_{s}^{(t)}$, because all non-zeros in $\Gbf$ have unit Gaussian distribution and the positive and negative part of this distribution is symmetric. Therefore, the projection $\Vbf_t\Vbf_t^\top$, which is a deterministic function of $\Gbf_{s}^{(t)}$, has the same distribution as the projection $\Pbf_i^\top\Vbf_t\Vbf_t^\top\Pbf_i$, which means
\begin{equation}\label{eqn:similar_projection}
\Ep[\Vbf_t\Vbf_t^\top]=\Ep[\Pbf_i^\top\Vbf_t\Vbf_t^\top\Pbf_i]=\Pbf_i^\top\Ep[\Vbf_t\Vbf_t^\top]\Pbf_i.
\end{equation}
By the structure of $\Pbf_i$, this means that after flipping the $i$-th column and the $i$-th row of $\Ep[\Vbf_t\Vbf_t^\top]$, the matrix stays the same. The only way that this can be true is that all entries on the $i$-th row and the $i$-th column are zero, except the diagonal entry.

\subsection{Property 2: all diagonal entries in $\Ep[\Vbf_t\Vbf_t^\top]$ are equal} Suppose $\Pbf_\pi$ is the $k\times k$ permutation matrix
\begin{equation}
\Pbf_\pi=\left[
\begin{matrix}
0 & 1 & 0 & \cdots & 0\\
0 & 0 & 1 & \cdots & 0\\
\vdots & \vdots & \vdots & \ddots & \vdots\\
0 & 0 & 0 & \cdots & 1\\
1 & 0 & 0 & \cdots & 0\\
\end{matrix}\right].
\end{equation}
It is an orthonormal matrix. When we multiply $\Gbf_{s}^{(t)}\Pbf_\pi$, we effectively do a cyclic shift on the columns of $\Gbf_{s}^{(t)}$ by pushing each column to its left (except for the left most column which is pushed to the right-most). Then, we show that the distribution of $\Gbf_{s}^{(t)}$ is again exactly the same as the distribution of $\Gbf_{s}^{(t)}\Pbf_\pi$. To prove this, we look at a specific random initialization of $\Gbf_{s}^{(t)}\Pbf_\pi$. Suppose $\Gbf_{s}^{(t)}$ occupies $a$ rows in $\Sbf_1$ and $b$ rows in $\Sbf_2$, where recall that $\Sbf_1$ and $\Sbf_2$ are square cyclic matrices (see Definition~\ref{def:cyclic}). Suppose the $a$ rows are the $i_1, i_2,\ldots,i_a$-th rows in $\Sbf_1$ and the $b$ rows are the $j_1, j_2,\ldots,j_b$-th rows in $\Sbf_2$. Then, after permuting the columns of $\Gbf_{s}^{(t)}$ by $\Pbf_\pi$, the sparsity pattern of $\Gbf_{s}^{(t)}\Pbf_\pi$ is the same as if the $(i_1-1, i_2-1,\ldots,i_a-1)(\text{mod } k)$-th rows in $\Sbf_1$ and the $(j_1-1, j_2-1,\ldots,j_b-1)(\text{mod } k)$-th rows in $\Sbf_2$ are chosen for a realization of $\Gbf_{s}^{(t)}$. This ensures that there exists a realization of $\Gbf_{s}^{(t)}$ which is exactly the same as the examined realization of $\Gbf_{s}^{(t)}\Pbf_\pi$. These two realizations have a one-to-one mapping because no other realizations of $\Gbf_{s}^{(t)}\Pbf_\pi$ will lead to the same realization\footnote{There is a subtle point here that when the cyclic rows of $\Sbf_1$ and $\Sbf_2$ are the same, the mapping is not one-to-one anymore. Therefore, $\Sbf_1$ and $\Sbf_2$ can be chosen to have different cyclic rows.} of $\Gbf_{s}^{(t)}$ and the pdfs of these two realizations are also the same, because all rows are selected uniformly and the all entries have the same unit Gaussian distribution. Based on this one-to-one mapping, we know that the distribution of $\Gbf_{s}^{(t)}$ and $\Gbf_{s}^{(t)}\Pbf_\pi$ are the same. Therefore, similar to \eqref{eqn:similar_projection}, we obtain
\begin{equation}\label{eqn:similar_projection_2}
\Ep[\Vbf_t\Vbf_t^\top]=\Pbf_\pi^\top\Ep[\Vbf_t\Vbf_t^\top]\Pbf_\pi.
\end{equation}
From Property 1 we have proved that all the off-diagonal entries in $\Ep[\Vbf_t\Vbf_t^\top]$ are 0. From \eqref{eqn:similar_projection_2}, we have that the cyclic shift on the diagonal also remains the same. This means all diagonal entries in $\Ep[\Vbf_t\Vbf_t^\top]$ are identical.

\subsection{Prove $\Ep[\Vbf_t\Vbf_t^\top]=(1-\delta_t)\I_k$} From Property 1 and Property 2, we can assume that $\Ep[\Vbf_t\Vbf_t^\top]=x\I_k$ for some constant $x$. Notice that on one hand,
\begin{equation}
\begin{split}
&\trace[\Ep[\Vbf_t\Vbf_t^\top]]=\Ep[\trace[\Vbf_t\Vbf_t^\top]]\\
=&\Ep[\trace[\Vbf_t^\top\Vbf_t]]=\Ep[\text{rank}(\Gbf_{s}^{(t)})].
\end{split}
\end{equation}
One the other hand,
\begin{equation}
\trace[\Ep[\Vbf_t\Vbf_t^\top]]=\trace[x\I_k]=xk,
\end{equation}
so we have $x=\frac{\Ep[\text{rank}(\Gbf_{s}^{(t)})]}{k}=1-\delta_t$.
\end{proof}

\section{Equivalence of Several Norms}
\subsection{For $G(N,p)$ the induced two-norm $\twonorm{\Abf}$ is close to $\rho(\Abf)$}\label{sec:norm_close}

We prove a lemma stating that for $G(N,p)$, the induced two-norm $\twonorm{\Abf}$ is close to $\rho(\Abf)$.
\begin{lemma}\label{lmm:norm_close}
For random graph from the Erd\"os-R\'enyi model $G(N,p)$, the column-normalized adjacency matrix $\Abf$ satisfies
\begin{equation}
\text{Pr}\left(\twonorm{\Abf}>\sqrt{\frac{1+\epsilon}{1-\epsilon}}\rho(\Abf)\right)<3Ne^{-\epsilon^2Np/8}.
\end{equation}
\end{lemma}
\begin{proof}
It is well known that the node degrees of Erd\"os-R\'enyi graphs concentrate at $Np$. For example, Theorem 4.1 in the online book chapter here \url{https://www.cs.cmu.edu/~avrim/598/chap4only.pdf} has the following theorem.
\begin{theorem}\label{thm:rand_graph}
Let $v$ be a vertex of the random graph $G(N,p)$. For $0<\alpha<\sqrt{Np}$
\begin{equation}
\text{Pr}(|Np-\text{deg}(v)|\ge \alpha\sqrt{Np})\le 3e^{-\alpha^2/8}.
\end{equation}
\end{theorem}
Therefore, by the union bound, with probability at least $1-3Ne^{-\epsilon^2Np/8}$, all nodes in the graph have degree within the range $(Np(1-\epsilon),Np(1+\epsilon))$. The matrix $\Abf$ is the column-normalized adjacency matrix, so its spectral radius $\rho(\Abf)=1$. By H\"older's inequality
\begin{equation}\label{eqn:2norm}
\twonorm{\Abf}_2\le \sqrt{\twonorm{\Abf}_1\twonorm{\Abf}_\infty}.
\end{equation}
The induced 1-norm $\twonorm{\Abf}_1$ is also the maximum absolute column sum and the induced infinity norm $\twonorm{\Abf}_\infty$ is also the maximum absolute row sum. Since $\Abf$ is the column-normalized, the induced 1-norm satisfies
\begin{equation}\label{eqn:1norm}
\twonorm{\Abf}_1=1.
\end{equation}
Since the degree of the graph is within range $(Np(1-\epsilon),Np(1+\epsilon))$ with high probability, the minimum normalization factor for a node $v$ when doing column normalization, i.e., the minimum column sum of the un-normalized adjacency matrix, is greater than $Np(1-\epsilon)$ with high probability. This means that with high probability, all non-zeros in $\Abf$ are smaller than $\frac{1}{Np(1-\epsilon)}$. Moreover, the maximum number of non-zeros of $\Abf$ in each row is also smaller than $Np(1+\epsilon)$ with high probability. This means that the maximum row sum $\twonorm{\Abf}_\infty$ satisfies the following with high probability
\begin{equation}\label{eqn:infnorm}
\twonorm{\Abf}_\infty\le \frac{Np(1+\epsilon)}{Np(1-\epsilon)}=\frac{1+\epsilon}{1-\epsilon}.
\end{equation}
From \eqref{eqn:2norm} to \eqref{eqn:infnorm}, with high probability (at least $1-3Ne^{-\epsilon^2Np/8}$),
\begin{equation}
\twonorm{\Abf}_2\le \sqrt{\frac{1+\epsilon}{1-\epsilon}}=\sqrt{\frac{1+\epsilon}{1-\epsilon}}\rho(\Abf).
\end{equation}
This completes the proof.
\end{proof}

\subsection{For $G(N,p)$ $\twonorm{\Abf}_\text{col}$ is close to $\rho(\Abf)$}
\begin{lemma}\label{lmm:norm_close_col}
For random graph from the Erd\"os-R\'enyi model $G(N,p)$, the column-normalized adjacency matrix $\Abf$ satisfies
\begin{equation}
\begin{split}
  &\text{Pr}\left(\twonorm{\Abf}_\text{col}>\sqrt{\left(1+\frac{k}{N}\right)\frac{1+\epsilon}{1-\epsilon}}\rho(\Abf)\right)\\
  <&3Ne^{-\epsilon^2Np/8}+3kNe^{-\epsilon^2Np/(8k)}.
\end{split}
\end{equation}
\end{lemma}
\begin{proof}
Notice that $\twonorm{\Abf}_\text{col}$ is defined as $\twonorm{\Abf}_\text{col}=\sqrt{k} \max_j\{\twonorm{\Abf_j}_2\}$, where $\Abf_1,\Abf_2\ldots,\Abf_k$ are the column blocks of the column-normalized adjacency matrix $\Abf$. For each $\Abf_j$, we will prove that with high probability, $\twonorm{\Abf_j}\approx \frac{1}{\sqrt{k}}\rho(\Abf)$. Then, using the union bound, we can prove that with high probability, $\sqrt{k} \max_j\{\twonorm{\Abf_j}_2\}\approx \rho(\Abf)$.

We divide the nodes of the graph into $k$ subsets according to the column-splitting $\Abf=[\Abf_1,\Abf_2,\ldots,\Abf_k]$. For an arbitrary $\Abf_j$, suppose its $i$-th row has $\text{deg}_{ij}$ non-zeros. Notice that $\text{deg}_{ij}$ is the degree of the $i$-th node in the subgraph induced by the $i$-th node and all nodes in the $j$-th subset of all nodes. Then, this subgraph is an Erd\"os-R\'enyi graph with random connection probability $p$ and node number $N/k+1$. Therefore, from Theorem~\ref{thm:rand_graph} and the union bound, with probability at least $1-3kNe^{-\epsilon^2(N/k+1)p/8}\ge 1-3kNe^{-\epsilon^2Np/(8k)}$, for all node $i$ and all subgraph $j$ the degree $\text{deg}_{ij}$ is in the range $((N/k+1)p(1-\epsilon),(N/k+1)p(1+\epsilon))$. Also, similar to the proof in Section~\ref{sec:norm_close}, with probability at least $1-3Ne^{-\epsilon^2Np/8}$, all nodes in the original graph have degree within the range $(Np(1-\epsilon),Np(1+\epsilon))$. Thus, with probability at least $1-3kNe^{-\epsilon^2Np/(8k)}-3Ne^{-\epsilon^2Np/8}$, the maximum row-sum of all $\Abf_j$ satisfies
\begin{equation}
\max_j \{\twonorm{\Abf_j}_\infty\}\le \frac{(N/k+1)p(1+\epsilon)}{Np(1-\epsilon)}=\left(\frac{1}{k}+\frac{1}{N}\right)\frac{1+\epsilon}{1-\epsilon}.
\end{equation}
By H\"older's inequality
\begin{equation}
\twonorm{\Abf_j}_2\le \sqrt{\twonorm{\Abf_j}_1\twonorm{\Abf_j}_\infty}.
\end{equation}
Since all $\Abf_j$ are column-normalized, the column sum $\twonorm{\Abf_j}_1=1,\forall j$. Thus, with probability at least $1-3kNe^{-\epsilon^2Np/(8k)}-3Ne^{-\epsilon^2Np/8}$,
\begin{equation}
\max_j \{\twonorm{\Abf_j}_2\} \le \sqrt{\left(\frac{1}{k}+\frac{1}{N}\right)\frac{1+\epsilon}{1-\epsilon}}.
\end{equation}
Since $\rho(\Abf)=1$, with high probability,
\begin{equation}
\twonorm{\Abf}_\text{col}=\sqrt{k} \max_j\{\twonorm{\Abf_j}_2\}\le \sqrt{\left(1+\frac{k}{N}\right)\frac{1+\epsilon}{1-\epsilon}}\rho(\Abf).
\end{equation}
\end{proof}
\end{document}